\documentclass{article}

\newcommand*{\extra}[1]{}

\usepackage[utf8]{inputenc}
\usepackage{amsfonts}
\usepackage{mathrsfs}
\usepackage{amsmath,amsthm,amssymb}
\usepackage{enumerate}
\usepackage[english]{babel}
\usepackage{graphicx}
\usepackage{nicefrac}
\usepackage[margin=3.8cm]{geometry}
\usepackage[usenames,dvipsnames]{xcolor}
\usepackage[colorlinks=true,urlcolor=Mahogany,linkcolor=Mahogany,citecolor=Mahogany,plainpages=false,pdfpagelabels]{hyperref}

\usepackage{authblk}

\usepackage{tikz}
\usetikzlibrary{shapes.geometric,plotmarks,backgrounds,fit,positioning,matrix,calc,circuits.logic.US}

\theoremstyle{plain}
\newtheorem{theorem}{Theorem}[section]
\newtheorem{lemma}[theorem]{Lemma}
\newtheorem{claim}[theorem]{Claim}
\newtheorem{corollary}[theorem]{Corollary}
\newtheorem{proposition}[theorem]{Proposition}

\theoremstyle{definition}
\newtheorem{definition}[theorem]{Definition}
\newtheorem{remark}[theorem]{Remark}

\newcommand*{\cD}{\mathcal{D}}
\newcommand*{\cE}{\mathcal{E}}

\newcommand*{\cI}{\mathcal{I}}

\newcommand*{\cM}{\mathcal{M}}

\newcommand*{\cS}{\mathcal{S}}
\newcommand*{\cT}{\mathcal{T}}

\newcommand*{\cX}{\mathcal{X}}
\newcommand*{\cY}{\mathcal{Y}}
\newcommand*{\cZ}{\mathcal{Z}}

\newcommand*{\eps}{\varepsilon}

\newcommand*{\id}{\mathrm{id}}
\newcommand*{\ident}{\mathrm{id}}
\newcommand*{\tr}{\mathrm{tr}}

\newcommand{\sket}[1]{{\ensuremath{\lvert#1\rangle}}}
\newcommand{\lket}[1]{{\ensuremath{\left\lvert#1\right\rangle}}}
\newcommand{\ket}[1]{\mathchoice{\lket{#1}}{\sket{#1}}{\sket{#1}}{\sket{#1}}}

\newcommand{\sbra}[1]{{\ensuremath{\langle#1\rvert}}}
\newcommand{\lbra}[1]{{\ensuremath{\left\langle#1\right\rvert}}}
\newcommand{\bra}[1]{\mathchoice{\lbra{#1}}{\sbra{#1}}{\sbra{#1}}{\sbra{#1}}}

\newcommand{\sketbra}[2]{{\ensuremath{\lvert #1\rangle\langle #2\rvert}}}
\newcommand{\lketbra}[2]{{\ensuremath{\left\lvert #1\middle\rangle\middle\langle #2\right\rvert}}}
\newcommand{\ketbra}[2]{\mathchoice{\lketbra{#1}{#2}}{\sketbra{#1}{#2}}{\sketbra{#1}{#2}}{\sketbra{#1}{#2}}}

\newcommand{\proj}[1]{\ketbra{#1}{#1}}

\newcommand{\Pos}{\mathrm{Pos}}
\newcommand{\Supp}{\mathrm{supp}}

\newcommand{\mbP}{\mathbb{P}}

\newcommand*{\freq}[1]{\mathsf{freq}(#1)}

\usepackage[title]{appendix}

\tikzstyle{porte} = [fill=blue!20, draw]

\begin{document}

\title{Entropy accumulation}

\date{}

\author[1,2,5]{Frédéric Dupuis}
\author[3]{Omar Fawzi}
\author[4]{Renato Renner}

\affil[1]{Université de Lorraine, CNRS, Inria, LORIA, F-54000 Nancy, France}
\affil[2]{Faculty of Informatics, Masaryk University, Brno, Czech Republic}
\affil[3]{Laboratoire de l'Informatique du Parall\'elisme, ENS de Lyon, France}
\affil[4]{Institute for Theoretical Physics, ETH Zurich, Zurich, Switzerland}
\affil[5]{Département d'informatique et de recherche opérationnelle, Université de Montréal, Montréal, Québec, Canada}

\maketitle

\begin{abstract}
 We ask the question whether entropy accumulates, in the sense that the operationally relevant total uncertainty about an $n$-partite system $A = (A_1, \ldots A_n)$ corresponds to the sum of the entropies of its parts~$A_i$. The  Asymptotic Equipartition Property implies that this is indeed the case to first order in~$n$ | under the assumption that the parts $A_i$ are identical and independent of each other. Here we show that entropy accumulation occurs more generally, i.e., without an independence assumption, provided one quantifies the uncertainty about the individual systems $A_i$ by the von Neumann entropy of suitably chosen conditional states. The  analysis of a large system can hence be reduced to the study of its parts.  This is relevant for applications. In device-independent cryptography, for instance, the approach yields essentially optimal security bounds valid for general attacks, as shown by Arnon-Friedman \emph{et al.}~\cite{AFRV19}. 
  \end{abstract}

\section{Introduction}

In classical information theory, the \emph{uncertainty} one has about a variable $A$ given access to side information $B$ can be operationally quantified by the number of bits one would need to learn, in addition to $B$, in order to reconstruct $A$. While this number generally fluctuates, it is | except with probability of order $\eps > 0$ | not larger than the \emph{$\eps$-smooth max-entropy}, $H_{\max}^\eps(A|B)_\rho$, evaluated for the joint distribution~$\rho$ of $A$ and $B$~\cite{RennerWolfb}.\footnote{There is some freedom in how to count the number of bits, but the statement always holds up to additive terms of the order $\log(1/\eps)$.}  Conversely, it is in the same way not smaller than the \emph{$\eps$-smooth min-entropy}, $H_{\min}^\eps(A|B)_\rho$. This may be summarised by saying that the number of bits needed to reconstruct $A$ from $B$ is with probability at least $1-O(\eps)$ contained in the interval
\begin{align} \label{eq_entropyinterval}
  I =  \bigl[H_{\min}^\eps(A|B)_\rho, H_{\max}^\eps(A|B)_\rho \bigr]  \ ,
\end{align}
whose boundaries are defined by the smooth entropies. We refer to Definition~\ref{def_smooth_min_max_entropy} below for a precise definition of these quantities.

This approach to quantifying uncertainty can be extended to the case where $A$ and $B$ are quantum systems. The conclusion remains the same: the operationally relevant uncertainty interval is $I$ as defined by~\eqref{eq_entropyinterval}. The only difference is that $\rho$ is now  a density operator, which describes the joint state of $A$ and $B$~\cite{RennerWolf,Renner05,Tom12}. 

Finding the boundaries of the interval $I$ is a central task of information theory. However, the smooth entropies of a large system $A$ are often difficult to calculate. It is therefore rather common to introduce certain assumptions to render this task more feasible. One extremely popular approach in standard information theory is to assume that the system consists of many mutually independent  and identically distributed (IID) parts. More precisely,  the \emph{IID Assumption} demands that the system be of the form  $A = A_1^n = A_1 \otimes \cdots \otimes A_n$, that the side information have an analogous structure $B = B_1^n = B_1 \otimes \cdots \otimes B_n$, and that the joint state of these systems be of the form $\rho_{A_1 B_1 \cdots A_n B_n} = \nu_{A B}^{\otimes n}$, for some density operator $\nu_{A B}$. A fundamental result from information theory, the \emph{Asymptotic Equipartition Property (AEP)}~\cite{Sha48} (see~\cite{TCR09} for the quantum version), then asserts that the uncertainty interval satisfies
\begin{align} \label{eq_AEP}
  I \subset \left[n \left(H(A | B)_{\nu} - \frac{c_\eps}{\sqrt{n}}\right), \,  n \left(H(A | B)_{\nu}  + \frac{c_\eps}{\sqrt{n}} \right) \right] \ ,
\end{align}
where $c_\eps$ is a constant (independent of $n$) and where $H(A|B)_{\nu}$ is the  conditional von Neumann entropy evaluated for the state $\nu_{A B}$. In other words, for large $n$, the operationally relevant total uncertainty one has about $A_1^n$ given $B_1^n$  is well approximated by $n H(A|B)_{\nu} = \sum_{i} H(A_i | B_i)_{\rho}$. In this sense, the entropy of the individual systems $A_i$ accumulates to the entropy of the total system $A_1^n$.\footnote{We note that the value of $c_{\eps}$ governing the second order term is well understood; see~\cite{TH13} and the book~\cite{Tom15} for more details on this.}

In this work, we generalise this statement to the case where the individual pairs $A_i B_i$ are no longer independent of each other, i.e., where the IID assumption does not hold. Without loss of generality one may think of the pairs $A_i  B_i$ as being generated by a sequence of processes $\cM_i$, as shown in Figure~\ref{fig:rho-structure}. Each process $\cM_i$ may pass information on to the next one via a ``memory'' register $R_i$. The state of the ``future'' pairs can thus depend on the ``past'' ones.\footnote{The IID assumption corresponds to the special case where the systems $R_{i}$ are trivial (ensuring the mutual independence of the pairs $A_i B_i$) and where the maps $\cM_i$ are all the same (ensuring that the state of each pair $A_i B_i$ is identical to all others).} The only assumption we make is that, given the side information $B_1^i$ generated until step~$i$, the systems $A_1^i$ are independent of the next piece of side information $B_{i+1}$. This is captured by the requirement that $A_1^{i} \leftrightarrow B_1^{i} \leftrightarrow B_{i+1}$ forms a quantum Markov chain.\footnote{The necessity of this condition is discussed in Appendix~\ref{app_Markov}.} Entropy accumulation is then the claim that
\begin{align} \label{eq_entropyaccumulation}
 I
 \subset  \left[\sum_{i =1}^n \left(  \inf_{\omega_{R_{i-1} R}} H(A_i | B_i R )_{\cM_i(\omega)} - \frac{c_\eps}{\sqrt{n}} \right), \, \sum_{i =1}^n \left(\sup_{\omega_{R_{i-1} R}} H(A_i | B_i R)_{\cM_i(\omega)} + \frac{c_\eps}{\sqrt{n}} \right) \right]
 \ ,
\end{align}
where, in the $i$th term of each sum, the infimum or supremum ranges over joint states $\omega_{R_{i-1} R}$ of the  memory $R_{i-1}$ and a system $R$ isomorphic to it, and the conditional von Neumann entropy  is evaluated for the state  $({\cM_i \otimes \cI_R})(\omega_{R_{i-1} R})$, abbreviated by  $\cM_i(\omega)$, which describes the output pair $A_i B_i$ generated by $\cM_i$ jointly with~$R$.

To illustrate~\eqref{eq_entropyaccumulation} it is useful to think of a communication scenario with two parties, Alice and Bob, who are receiving information $A_1^n$ and $B_1^n$, respectively. Suppose that a source with memory  $R_i$ generates this information sequentially in $n$ steps, described by maps $\cM_i$ as depicted in Figure~\ref{fig:rho-structure}. Suppose furthermore that Bob would like to infer all $n$ values $A_i$ (which, for the purpose of this example, we  assume to be classical). As discussed above, for this he would require $N$ additional classical bits from Alice, where $N$ fluctuates (up to probability~$\varepsilon$)  within an interval $I$ with boundaries given by the entropies $H_{\min}^{\varepsilon}(A_1^n|B_1^n)$ and $H_{\max}^{\varepsilon}(A_1^n|B_1^n)$, which quantify Bob's uncertainty about $A_1^n$. While these entropies depend on the joint state $\rho_{A_1^n B_1^n}$ of the entire information generated by the source over all $n$ steps, they can, according to~\eqref{eq_entropyaccumulation}, be lower (or upper) bounded by a sum of terms that merely depend on  the individual steps~$\cM_i$. Specifically, the minimum (or maximum) number $N$ of bits that Alice needs to send to Bob so that he can infer her values~$A_i$  grows for each such value by the von Neumann entropy $H(A_i|B_i R)$, minimised (or maximised) over all possible states the memory $R_{i-1}$ could have been in right before the pair $A_i B_i$ was produced, and conditioned on~$B_i$ as well as any information $R$ about this memory.\footnote{In the special case where the source is  IID, the memory register is trivial and no minimisation (or maximisation) is necessary. Expression~\eqref{eq_entropyaccumulation} then reduces to~\eqref{eq_AEP}, and one retrieves the (well known) result that the number of bits that Alice needs to communicate to Bob per value $A_i$ is, up to second order terms, given by $H(A_i|B_i)$.}   

\begin{figure}
    \begin{center}
    \begin{tikzpicture}[thick]
        \tikzstyle{porte} = [draw=blue!50, fill=blue!20]
        \draw
            (0, 0) node[porte] (m1) {$\cM_1$}
            ++(2, 0) node[porte] (m2) {$\cM_2$}
            ++(2, 0) node (dotdotdot) {$\cdots$}
            ++(2, 0) node[porte] (mn) {$\cM_n$}
            (m1) ++(-.5, -1.2) node (a1) {$A_1$}
            (m1) ++(.5, -1.2) node (b1) {$B_1$}
            (m2) ++(-.5, -1.2) node (a2) {$A_2$}
            (m2) ++(.5, -1.2) node (b2) {$B_2$}
            (mn) ++(-.5, -1.2) node (an) {$A_n$}
            (mn) ++(.5, -1.2) node (bn) {$B_n$}
            ;
       \draw 
            (m1) ++(-1.5, 0) edge[->] node[midway, above] {$R_0$} (m1)
            (m1) edge[->] node[midway, above] {$R_1$} (m2)
            (m2) edge[->] node[midway, above] {$R_2$} (dotdotdot)
            (dotdotdot) edge[->] node[midway, above] {$R_{n-1}$} (mn)
            (m1) edge[->] (a1)
            (m1) edge[->] (b1)
            (m2) edge[->] (a2)
            (m2) edge[->] (b2)
            (mn) edge[->] (an)
            (mn) edge[->] (bn)
            ;
    \end{tikzpicture}
    \end{center}
    \caption{Circuit diagram illustrating the decomposition of states $\rho_{A_1^n B_1^n}$ relevant for our main theorem. One starts with a state $\rho^0_{R_0}$, and each of the pairs $A_i B_i$ is generated sequentially, one after the other, by the process $\cM_i$. The map $\cM_i$ takes as input a state on $R_{i-1}$ and outputs a state on $R_{i} \otimes A_i \otimes B_i$.}
    \label{fig:rho-structure}
\end{figure}

The main result we derive in this work is actually a bit more general than~\eqref{eq_entropyaccumulation}, allowing one to take into account global information about the statistics of $A_1^n$ and $B_1^n$. This is relevant for applications. In quantum key distribution, for instance,  $\cM_i$ models the generation of the $i$th bit of the raw key. However, in this cryptographic scenario, $\cM_i$ can depend on the attack strategy of an adversary, and is thus partially unknown. Hence, in order to bound the entropy  (which characterises an adversary's uncertainty) of the raw key bits, one must as well take into account global statistical properties. These are inferred by tests carried out by the quantum key distribution protocol on a small sample of the generated bits. To incorporate such statistical information in the analysis, we consider for each~$i$ an additional classical value $X_i$ derived from $A_i$ and $B_i$, as depicted by Figure~\ref{fig:rho-structure-X}. Specifically, $X_i$  shall tell us whether position $i$ was included in the statistical test, and if so, the outcome of the test performed at step $i$. For this extended scenario, \eqref{eq_entropyaccumulation} still holds, but now the infimum and supremum are taken over a restricted set, containing only those states~$\omega$ for which the resulting probability distribution on $X_i$ corresponds to the observed statistics. 

\begin{figure}
    \begin{center}
    \begin{tikzpicture}[thick]
        \tikzstyle{porte} = [draw=blue!50, fill=blue!20]
        \draw
            (0, 0) node[porte] (m1) {$\cM_1$}
            ++(2, 0) node[porte] (m2) {$\cM_2$}
            ++(2, 0) node (dotdotdot) {$\cdots$}
            ++(2, 0) node[porte] (mn) {$\cM_n$}
            (m1) ++(-.5, -1.2) node (a1) {$A_1$}
            (m1) ++(.5, -1.2) node (b1) {$B_1$}
            (m1) ++(0, -2) node (x1) {$X_1$}
            (m2) ++(-.5, -1.2) node (a2) {$A_2$}
            (m2) ++(.5, -1.2) node (b2) {$B_2$}
            (m2) ++(0, -2) node (x2) {$X_2$}
            (mn) ++(-.5, -1.2) node (an) {$A_n$}
            (mn) ++(.5, -1.2) node (bn) {$B_n$}
            (mn) ++(0, -2) node (xn) {$X_n$}
            ;
       \draw 
            (m1) ++(-1.5, 0) edge[->] node[midway, above] {$R_0$} (m1)
            (m1) edge[->] node[midway, above] {$R_1$} (m2)
            (m2) edge[->] node[midway, above] {$R_2$} (dotdotdot)
            (dotdotdot) edge[->] node[midway, above] {$R_{n-1}$} (mn)
            (m1) edge[->] (a1)
            (m1) edge[->] (b1)
            (m2) edge[->] (a2)
            (m2) edge[->] (b2)
            (mn) edge[->] (an)
            (mn) edge[->] (bn)
            (a1) edge[->, dotted] (x1)
            (b1) edge[->, dotted] (x1)
            (a2) edge[->, dotted] (x2)
            (b2) edge[->, dotted] (x2)
            (an) edge[->, dotted] (xn)
            (bn) edge[->, dotted] (xn)
            ;
    \end{tikzpicture}
    \end{center}
    \caption{Circuit diagram illustrating the decomposition of states $\rho_{A_1^n B_1^n X_1^n}$ relevant for the full version of our main theorem, which can take into account statistical information $X_1^n$. The individual pieces   $X_i$ of this statistical information are classical values that can be determined from $A_i$ and $B_i$ without disturbing them. When $A_i$ and $B_i$ are themselves classical, this means that $X_i$ is a deterministic function of $A_i$ and $B_i$. For a precise definition in the general case we refer to Section~\ref{sec_accumulation}.}
    \label{fig:rho-structure-X}
\end{figure}

Entropy accumulation has a number of theoretical and practical implications.  For example, it serves as a technique to turn cryptographic security proofs that were restricted to collective attacks to security proofs against general attacks. This application is demonstrated in~\cite{AFRV19} for the case of a fully device-independent quantum key distribution and a randomness expansion protocol.
The resulting security bounds are essentially tight, implying that device-independent cryptography is  possible with state-of-the-art technology. To illustrate the basic ideas behind such applications, we will present two concrete examples in more detail. The first is a proof of security of a variant of the \emph{E91 Quantum Key Distribution} protocol. This new security proof has two advantages. First, its structure is modular and it may therefore be adapted to other cryptographic schemes (see also the discussion in Section~\ref{sec_conclusions}). In addition, it achieves a strong level of security where no assumption is made on Bob's devices. This is sometimes referred to as one-sided measurement device independence and this level of security was partially achieved in~\cite{TR11} (they used a memoryless devices assumption which we do not need) and later fully in~\cite{TFKW13} though with sub-optimal rates.
 The second example is the derivation of an upper bound on the fidelity achievable by \emph{Fully Quantum Random Access Codes}.

The proof of the main result, Eq.~\eqref{eq_entropyaccumulation}, has a similar structure as the proof of the Quantum Asymptotic Equipartition Property~\cite{TCR09}, which we can retrieve as a special case (see Corollary~\ref{cor_QAEP}). The idea is to first bound the smooth entropy of the entire sequence $A_1^n$ conditioned on $B_1^n$ by a conditional R\'enyi entropy of order $\alpha$, then decompose this entropy into a sum of conditional R\'enyi entropies for the individual terms $A_i$, and finally bound these  in terms of  von Neumann entropies. However, in contrast to previous arguments, we use a recently introduced version of conditional R\'enyi entropies, termed ``sandwiched R\'enyi entropies''~\cite{WWY13,MDSFT13}.  For these entropies, we derive a novel chain rule, which forms the core technical part of our proof.  In addition, some of the concepts used in this work generalise techniques proposed in the recent security proofs for device-independent cryptography presented in~\cite{MS14a,MS14b}. 
In particular, the dominant terms of the lower bound on the amount of  randomness obtained in~\cite{MS14b}, called \emph{rate curves}, are similar to the \emph{tradeoff functions} considered here (cf.\ Definition~\ref{def_tradeoff}).\footnote{While the tradeoff functions considered in this work are defined in terms of conditional von Neumann entropies, the rate curves of~\cite{MS14b} are equal to a difference of $(1+\varepsilon)$-R\'enyi entropies (see \cite[Section~6]{MS14b}). The latter cannot be larger than the tradeoff functions, which yield asymptotically optimal randomness extraction rates (as shown in~\cite{AFRV19}).}

\paragraph{Paper organisation: } We begin with preliminaries and notation in Section~\ref{sec:preliminaries}.  Section~\ref{sec_chain} is devoted to the central technical ingredient of our argument, a chain rule for R\'enyi entropies. The main result, the theorem on entropy accumulation, is then stated and proved in Section~\ref{sec_accumulation}. In Section~\ref{sec_applications} we present the two sample applications mentioned above, before concluding with  remarks and suggestions for future work in Section~\ref{sec_conclusions}.

\section{Preliminaries}\label{sec:preliminaries}

\subsection{Notation} \label{sec_Notation}
In the table below, we summarise some of the notation used throughout the paper:
\begin{center}
    \begin{tabular}{|c|l|}
        \hline
        \emph{Symbol} & \multicolumn{1}{c|}{\emph{Definition}}\\
        \hline
        $A, B, C, \dots$ & Quantum systems, and their associated Hilbert spaces\\
        $\mathcal{L}(A,B)$ & Set of linear operators from $A$ to $B$\\
        $\mathcal{L}(A)$ & $\mathcal{L}(A,A)$\\
        $X_{AB}$ & Operator in $\mathcal{L}(A \otimes B)$\\
        $X_{B \leftarrow A}$ & Operator in $\mathcal{L}(A, B)$\\
        $\mathrm{D}(A)$ & Set of normalised density operators on $A$\\
        $\mathrm{D}_{\leqslant}(A)$ & Set of sub-normalised density operators on $A$\\
        $\Pos(A)$ & Set of positive semidefinite operators on $A$\\
        $X^{-1}$ for $X \in \Pos(A)$ & Generalised inverse, such that $XX^{-1}X = X$ holds\\
        $X_A \geqslant Y_A$ & $X_A - Y_A \in \Pos(A)$\\
        $A_i^j$ (with $j \geqslant i$) & Given $n$ systems $A_1,\dots,A_n$, this is a shorthand for $A_i,\dots,A_j$\\
        $A^n$ & Often used as shorthand for $A_1,\dots,A_n$\\
        $\log(x)$ & Logarithm of $x$ in base 2\\
        \hline
    \end{tabular}
\end{center}

Throughout the paper, we restrict ourselves to finite-dimensional Hilbert spaces. Furthermore, we use the following notation for classical-quantum states $\rho_{X A} \in \mathrm{D}(X \otimes A)$ with respect to the basis $\{\ket{x}\}_{x \in \cX}$ of the system $X$. For any $x \in \cX$, we let $\rho_{A,x} = \bra{x} \rho_{AX} \ket{x}$ so that $\rho_{XA} = \sum_{x \in \cX}  \proj{x} \otimes \rho_{A, x}$. To refer to the conditional state, we write $\rho_{A|x} = \frac{\rho_{A,x}}{\tr(\rho_{A,x})}$. An event $\Omega \subseteq \cX$ in this paper refers to a subset of $\cX$ and we can similarly define $\rho_{XA|\Omega} = \frac{1}{\rho[\Omega]} \sum_{x \in \Omega} \proj{x} \otimes \rho_{A,x} $, where we introduced the notation $\rho[\Omega] = \sum_{x \in \Omega} \tr(\rho_{A, x})$. We also use the usual notation for the partial trace for conditional states, e.g., $\rho_{XA|\Omega} = \tr_{B}(\rho_{XAB|\Omega})$.

For a density operator $\rho_{A B} \in \mathrm{D}(A \otimes B)$ on a bipartite Hilbert space $A \otimes B$ we define the operator\footnote{Note that this operator is well defined, for the support of $\rho_{A B}$ is contained in the support of $\id_A \otimes \rho_B$.}
\begin{align*}
  \rho_{A | B} = (\id_A \otimes \rho_B)^{-\frac{1}{2}} \rho_{A B} (\id_A \otimes \rho_B)^{-\frac{1}{2}}  \ ,
\end{align*}
which may be interpreted as the \emph{state of $A$ conditioned on $B$}, analogous to a conditional probability distribution. This operator was previously defined and studied in~\cite{AKMS05,Lei07}. In the following, we will usually drop identity operators from the notation when they are clear from the context. We would thus write, for instance, 
\begin{align*}
  \rho_{A | B} = \rho_B^{-\frac{1}{2}} \rho_{A B} \rho_B^{-\frac{1}{2}}  \ .
\end{align*}

\begin{remark}
  Let $A$ and $\bar{A}$ be two isomorphic Hilbert spaces with orthonormal bases $\{\ket{i}_A\}$ and $\{\ket{i}_{\bar{A}}\}$ and define
  \begin{align*}
    \ket{\Theta} = \sum_{i} \ket{i}_A \otimes \ket{i}_{\bar{A}} \ .
  \end{align*}
  Then any trace-non-increasing map $\cM = \cM_{B \leftarrow \bar{A}}$  from $\mathcal{L}(\bar{A})$ to $\mathcal{L}(B)$ can be represented as a ``conditional state'' (also known as the Choi-Jamiolkowski state) $M_{B | A}$ on $A \otimes B$ with the property that 
  \begin{align} \label{eq_condsubstate}
      M_{B | A} \geqslant 0 \qquad \text{and} \qquad \tr_B[M_{B|A}] \leqslant \id_A
  \end{align}
  and such that 
    \begin{align} \label{eq_mapMrepresentation}
    \cM(\rho_A^{\frac{1}{2}} \proj{\Theta} \rho_A^{\frac{1}{2}}) = \rho_A^{\frac{1}{2}} M_{B|A} \rho_A^{\frac{1}{2}} \qquad \forall \rho_A 
  \end{align}
  holds.   Specifically, for any map $\cM$ one may define
  \begin{align} \label{eq_mapMdefinition}
    M_{B|A} = \cM(\proj{\Theta})  \ ; 
  \end{align}
  it is then straightforward to verify the properties above. 

  Conversely, for any $M_{B | A}$ such that~\eqref{eq_condsubstate} holds the map defined by
  \begin{align*}
    \cM(\ket{i}_{\bar{A}} \bra{j}_{\bar{A}}) = \bra{i}_A M_{B|A} \ket{j}_A
  \end{align*}
  satisfies~\eqref{eq_mapMdefinition} and hence~\eqref{eq_mapMrepresentation}. It is also easy to verify that it is completely positive and trace non-increasing.
\end{remark}

We mention here a slight abuse of terminology: for a completely positive map $\cM_{B \leftarrow A}$ from $\mathcal{L}(A)$ to $\mathcal{L}(B)$, we often use a shorthand to indicate the systems it acts on and simply say that it maps $A$ to $B$.

\subsection{Background on quantum Markov chains}
The concept of quantum Markov chains will be used throughout the paper, and here we give some relevant basic facts about them. Let $\{a_j\}_{j \in J}$ and $\{c_j\}_{j \in J}$ be families of Hilbert spaces and let $B$ be a Hilbert space such that\footnote{$\oplus$ denotes the orthogonal direct sum.}
\begin{align} \label{eq_decomposition}
  B \cong \bigoplus_{j \in J} a_j \otimes c_j \ ,
\end{align}
holds. Let us furthermore denote by~$V = \bigoplus_{j \in J} V_{a_j c_j \leftarrow B}$  the corresponding isomorphism. It is convenient to treat $\bigoplus_j a_j \otimes c_j$ as a subspace of the product  $a \otimes c$ of the spaces
\begin{align*}
  a = \bigoplus_{j \in J} a_j \qquad \text{and} \qquad
  c = \bigoplus_{j \in J} c_j \ .
\end{align*}
 The mapping $V$ may then be viewed as an embedding of $B$ into $a \otimes c$. Given a density operator $\rho_B$, we denote by $\rho_{a c}$ the density operator $V \rho_B V^{\dagger}$. More generally, for a multi-partite density operator $\rho_{A B}$, we write $\rho_{A a c}$ for $V \rho_{A B} V^{\dagger}$. Furthermore, for any $j \in J$, we denote by $\rho_{A a_j c_j}$ the projection of $\rho_{A a c}$ onto the subspace defined by $a_j \otimes c_j$, i.e., 
\begin{align}  \label{eq_projop}
 \rho_{A a_j c_j} = V_{a_j c_j \leftarrow B} \rho_{A B} V_{a_j c_j \leftarrow B}^{\dagger} \ .
 \end{align}

A tri-partite density operator $\rho_{A B C}$ is said to obey the \emph{Markov chain condition} $A \leftrightarrow B \leftrightarrow C$ if there exists a decomposition of $B$ of the form~\eqref{eq_decomposition} such that
\begin{align} \label{eq_Markovcondtrad}
  \rho_{A B C} \cong \rho_{A a c C} = \bigoplus_{j \in J}  q_j \hat{\rho}_{A a_j} \otimes \hat{\rho}_{c_j C} 
\end{align}
where $\{q_j\}_{j \in J}$ is a probability distribution and $\{\hat{\rho}_{A a_j}\}_{j \in J}$ and $\{\hat{\rho}_{c_j C}\}_{j \in J}$ are families of density operators~\cite{Pet88,KI02,HJPW04}. It follows from this decomposition that a state $\rho_{ABC}$ obeying the Markov chain condition can be reconstructed from $\rho_{AB}$ with a map $\cT_{BC \gets B}$ acting only on $B$~\cite{Pet88}:
\begin{align}\label{eq_recovery}
\rho_{A B C} = \cI_{A} \otimes \cT_{BC \gets B}(\rho_{AB}) \ .
\end{align}
Another useful characterization of the Markov chain condition for $\rho_{ABC}$ is given by the entropic equality $I(A:C|B)_{\rho} = 0$~\cite{Pet88,KI02,HJPW04}. The conditional mutual information is defined as $I(A:C|B)_{\rho} = H(AB)_{\rho} + H(BC)_{\rho} - H(B)_{\rho} - H(ABC)_{\rho}$ where $H(A)_{\rho} = -\tr(\rho_{A} \log \rho_{A})$ is the von Neumann entropy.
   
\subsection{Entropic quantities}

The formulation of the main claim refers to smooth entropies, which can be defined as follows. 

\begin{definition}
\label{def_smooth_min_max_entropy}
  For any density operator $\rho_{A B}$ and for $\eps \in [0,1]$ the \emph{$\eps$-smooth min- and max-entropies of $A$ conditioned on $B$} are 
  \begin{align*}
    H_{\min}^\eps(A|B)_{\rho} & =  - \log \inf_{\tilde{\rho}_{A B}} \inf_{\sigma_B}  \left\| \tilde{\rho}_{A B}^{\frac{1}{2}} \sigma_B^{-\frac{1}{2}} \right\|_{\infty}^2 \\
    H_{\max}^\eps(A|B)_{\rho} & = \log \inf_{\tilde{\rho}_{A B}}  \sup_{\sigma_B}  \left\| \tilde{\rho}_{A B}^{\frac{1}{2}} \sigma_B^{\frac{1}{2}} \right\|_1^2 \ ,
  \end{align*}
  respectively, where $\tilde{\rho}$ is any non-negative operator with trace at most $1$ that is $\eps$-close to $\rho$ in terms of the purified distance\footnote{The purified distance is defined as $P(\rho, \tilde{\rho}) = \sqrt{1 - \big(\tr \big| \sqrt{\rho} \sqrt{\tilde{\rho}} \big| \big)^2 }$ whenever either $\rho$ or $\tilde{\rho}$ is normalized.}~\cite{TCR10,Tom12}, and where $\sigma_B$ is any density operator on $B$.
\end{definition}   
   
The proof we present here heavily relies on the sandwiched relative R\'enyi entropies introduced in~\cite{WWY13,MDSFT13}.  These relative entropies can be used to define a conditional entropy.\footnote{We note that there are at least two common variants for how to define a conditional entropy based on a relative entropy. We refer to Appendix~\ref{app_sandwich} for more details. }

\begin{definition} \label{def_sandwichedentropy}
  For any density operator $\rho_{A B}$ and for $\alpha \in (0, 1) \cup (1, \infty)$ the \emph{sandwiched $\alpha$-R\'enyi entropy of $A$ conditioned on $B$} is defined as
  \begin{align*}
    H_\alpha(A|B)_{\rho} = - \frac{1}{\alpha'} \log \left\| \rho_{A B}^{\frac{1}{2}} \rho_B^{\frac{-\alpha'}{2}}  \right\|_{2 \alpha}^2 \ ,
      \end{align*}
      where $\alpha' = \frac{\alpha - 1}{\alpha}$ and where $\| X \|_\alpha = \tr\bigl( (X^{\dagger} X)^{\frac{\alpha}{2}} \bigr)^{\frac{1}{\alpha}}$. Note that $\alpha'$ is the inverse of the H\"older conjugate of $\alpha$.
\end{definition}

We note that, while the function $X \mapsto \| X \|_\alpha$ is a norm for $\alpha \geqslant 1$, this is not the case for $\alpha < 1$ since it does not satisfy the triangle inequality. Some key properties of this function are summarised in Appendix~\ref{app_alpha}.  Using them, the sandwiched R\'enyi entropies may be rewritten as
\begin{align*}
    H_\alpha(A|B)_{\rho}  
    & =  \frac{\alpha}{1-\alpha} \log \left\| \rho_B^{\frac{1-\alpha}{2\alpha}} \rho_{A B}  \rho_B^{\frac{1-\alpha}{2\alpha}}  \right\|_\alpha  \\
    & =    \frac{1}{1-\alpha} \log \tr\left( \bigl( \rho_B^{\frac{1-\alpha}{2\alpha}} \rho_{A B} \rho_B^{\frac{1-\alpha}{2\alpha}} \bigr)^{\alpha} \right) \ .
  \end{align*}

It turns out that there are multiple ways of defining conditional entropies from relative entropies. Another variant that will be needed in this work is the following:
\begin{definition} \label{def_sandwichedentropy_uparrow}
  For any density operator $\rho_{A B}$ and for $\alpha \in (0, 1) \cup (1, \infty)$, we define
  \begin{align*}
      H^{\uparrow}_\alpha(A|B)_{\rho} = - \inf_{\sigma_B} \frac{1}{\alpha'} \log \left\| \rho_{A B}^{\frac{1}{2}} \sigma_B^{\frac{-\alpha'}{2}}  \right\|_{2 \alpha}^2 \ ,
      \end{align*}
  where the infimum is over all sub-normalised density operators on $B$.
\end{definition}

  Other relevant facts about the sandwiched R\'enyi entropy and the corresponding notion of relative entropy can be found in Appendix \ref{app_sandwich}.
  
\section{Chain rule for R\'enyi entropies} \label{sec_chain}

As explained in the introduction, our main result can be regarded as a generalisation of the  Quantum Asymptotic Equipartition Property~\cite{TCR09}, corresponding to~\eqref{eq_AEP}.  The approach used for the proof of the latter is to bound both the smooth min-entropy and the von Neumann entropy by R\'enyi entropies with an appropriate parameter $\alpha$. The IID assumption is then used to decompose the R\'enyi entropy into a sum of $n$ terms. However, since our main claim, Eq.~\eqref{eq_entropyaccumulation}, is supposed to hold for general non-IID states,  we do not have this luxury, and we must somehow decompose the R\'enyi entropy into $n$ terms using other means. The tool we will use for this purpose is a chain rule for R\'enyi entropies, which we present as a separate theorem (Theorem~\ref{thm_chainq_exact}). We start by stating a more general version that will be useful in the proof of the main theorem.

\begin{lemma} \label{lem_chainprep}
  Let $\rho_{A_1 A_2 B}$ and $\sigma_B$ be density operators and let $\alpha \in (0, \infty)$. Then
     \begin{align*} 
             D_{\alpha}(\rho_{A_1 B} \| \id_{A_1} \otimes  \sigma_B) - D_\alpha(\rho_{A_1 A_2 B}\| \id_{A_1 A_2} \otimes  \sigma_B) &=  H_{\alpha}(A_2 | A_1 B)_{\nu} \ ,
  \end{align*}
where 
\begin{align}
\label{eq_alphastate}
\nu_{A_1A_2B} = \nu^{\frac12}_{A_1B} \rho_{A_2|A_1B} \nu^{\frac12}_{A_1B} \: \text{ with } \: \nu_{A_1B} = \frac{\left(\rho_{A_1B}^{\frac12} \sigma^{\frac{1-\alpha}{\alpha}}_{B} \rho_{A_1B}^{\frac12}\right)^{\alpha}}{\tr\left( \rho_{A_1B}^{\frac12} \sigma^{\frac{1-\alpha}{\alpha}}_{B} \rho_{A_1B}^{\frac12}\right)^{\alpha}} \ .
\end{align}
\end{lemma}

We note that $\nu_{A_1B} = \tr_{A_2}(\nu_{A_1A_2B})$, which justifies the notation. 

\begin{proof}
When $\alpha = 1$, this equality follows directly from the definition of the entropies.  To prove the equality for $\alpha \in (0,1) \cup (1, \infty)$ we consider a purification $\ket{\psi}_{A_1A_2BE}$ of $\rho_{A_1A_2B}$. Using Lemma~\ref{lem_Halphavector} and setting $\alpha' = \frac{\alpha-1}{\alpha}$ we have
\begin{align*}
&D_\alpha(\rho_{A_1 A_2 B}\| \id_{A_1 A_2} \otimes \sigma_B) \\
&=  \sup_{\tau_E} \frac{1}{\alpha'} \log \left\| \sigma_{B}^{-\frac{\alpha'}{2}} \otimes \tau_{E}^{\frac{\alpha'}{2}} \ket{\psi} \right\|^2 \\
&= \sup_{\tau_E} \frac{1}{\alpha'} \log \left\| \left(\sigma_{B}^{-\frac{\alpha'}{2}} \rho_{A_1B}^{\frac{1}{2}} \otimes \tau_{E}^{\frac{\alpha'}{2}}\right) \rho_{A_1B}^{-\frac{1}{2}} \ket{\psi} \right\|^2 \\
&= \sup_{\tau_E} \frac{1}{\alpha'} \log \tr\left( \left(\sigma_{B}^{-\frac{\alpha'}{2}} \rho_{A_1B}^{\frac{1}{2}} \otimes \tau_{E}^{\frac{\alpha'}{2}} \right) \rho_{A_1B}^{-\frac{1}{2}} \proj{\psi} \rho_{A_1B}^{-\frac{1}{2}} \left(\rho_{A_1B}^{\frac{1}{2}} \sigma_{B}^{-\frac{\alpha'}{2}}   \otimes \tau_{E}^{\frac{\alpha'}{2}} \right) \right) \\
&= \sup_{\tau_E} \frac{1}{\alpha'} \log \tr\left( \left(\rho_{A_1B}^{\frac{1}{2}} \sigma_{B}^{-\alpha'} \rho_{A_1B}^{\frac{1}{2}} \otimes \tau_{E}^{\frac{\alpha'}{2}} \right) \rho_{A_1B}^{-\frac{1}{2}} \proj{\psi} \rho_{A_1B}^{-\frac{1}{2}} \left(\id_{A_1B} \otimes \tau_{E}^{\frac{\alpha'}{2}} \right) \right) \ .
\end{align*}
By the definition of $\nu_{A_1B}$, we get
\begin{align*}
&D_\alpha(\rho_{A_1 A_2 B}\| \id_{A_1 A_2} \otimes \sigma_B) \\
&= \frac{1}{\alpha'} \log \left(\tr\left( \rho_{A_1B}^{\frac{1}{2}} \sigma^{-\alpha'}_{B} \rho_{A_1B}^{\frac{1}{2}}\right)^{\alpha}\right)^{\frac{1}{\alpha}} \\
&\;\;\;\;\;\;\; + \sup_{\tau_E} \frac{1}{\alpha'} \log \tr\left( \left(\nu_{A_1B}^{\frac{1}{\alpha}} \otimes \tau_{E}^{\frac{\alpha'}{2}}\right) \rho_{A_1B}^{-\frac{1}{2}} \proj{\psi} \rho_{A_1B}^{-\frac{1}{2}} \left(\id_{A_1A_2B} \otimes \tau_{E}^{\frac{\alpha'}{2}} \right) \right) \\
&= D_\alpha(\rho_{A_1 B} \| \id_{A_1} \otimes \sigma_B) \\
&\;\;\;\;\;\;\; + \sup_{\tau_E} \frac{1}{\alpha'} \log \tr\left( \left(\nu_{A_1B}^{-\frac{\alpha'}{2}} \otimes \tau_{E}^{\frac{\alpha'}{2}} \right) \nu_{A_1B}^{\frac{1}{2}} \rho_{A_1B}^{-\frac{1}{2}} \proj{\psi} \rho_{A_1B}^{-\frac{1}{2}} \nu_{A_1B}^{\frac{1}{2}} \left(\nu_{A_1B}^{-\frac{\alpha'}{2}} \otimes \tau_{E}^{\frac{\alpha'}{2}} \right) \right) \\
&= D_\alpha(\rho_{A_1 B} \| \id_{A_1} \otimes \sigma_B) + \sup_{\tau_E} \frac{1}{\alpha'} \log \left\| \nu_{A_1B}^{-\frac{\alpha'}{2}} \otimes \tau_{E}^{\frac{\alpha'}{2}} \ket{\nu} \right\|^2 \ ,
\end{align*}
where we defined the pure state $\ket{\nu}_{A_1A_2BE} = \nu_{A_1B}^{\frac12} \rho_{A_1B}^{-\frac12} \ket{\psi}$, which is a purification of $\nu_{A_1A_2B}$.
To conclude we use the fact that $\nu_{A_1B} = \tr_{A_2}(\nu_{A_1A_2B})$ and Lemma~\ref{lem_Halphavector}.
\end{proof}

By choosing $\sigma_B = \rho_{B}$ in Lemma~\ref{lem_chainprep}, we directly obtain a chain rule for the R\'enyi entropies:

\begin{theorem} \label{thm_chainq_exact}
  Let $\rho_{A_1 A_2 B}$ be a density operator and let $\alpha \in (0, \infty)$. Then
     \begin{align} 
             H_\alpha(A_1 A_2 | B)_{\rho} &= H_{\alpha}(A_1 | B)_{\rho} + H_{\alpha}(A_2 | A_1 B)_{\nu} \ ,
  \end{align}
where 
\begin{align}
\label{eq_alphastate2}
\nu_{A_1A_2B} = \nu^{\frac12}_{A_1B} \rho_{A_2|A_1B} \nu^{\frac12}_{A_1B} \: \text{ with } \: \nu_{A_1B} = \frac{\left(\rho_{A_1B}^{\frac12} \rho^{\frac{1-\alpha}{\alpha}}_{B} \rho_{A_1B}^{\frac12}\right)^{\alpha}}{\tr\left( \rho_{A_1B}^{\frac12} \rho^{\frac{1-\alpha}{\alpha}}_{B} \rho_{A_1B}^{\frac12}\right)^{\alpha}} \ .
\end{align}
\end{theorem}
One drawback of the above result is that we are seldom interested in the particular state $\nu$ defined in the theorem statement. It is therefore generally more useful to present the result in a slightly weaker form, where the state $\nu$ is chosen to be the worst case over an appropriate class of density operators. When $\rho$ obeys the Markov chain condition $A_1 \leftrightarrow B_1 \leftrightarrow B_2$, we obtain the following result.

\begin{theorem}
\label{thm_chainq}
Let $\rho_{A_1 B_1 A_2 B_2}$ be a density operator such that  the Markov chain condition $A_1 \leftrightarrow B_1 \leftrightarrow B_2$ holds and let $\alpha \in (0, \infty)$. Then
   \begin{align} 
          \inf_{\nu} H_{\alpha}(A_2 | B_2 A_1 B_1)_{\nu} 
  \leqslant H_\alpha(A_1 A_2 | B_1 B_2)_{\rho} - H_{\alpha}(A_1 | B_1)_{\rho} \leqslant \sup_{\nu} H_{\alpha}(A_2 | B_2 A_1 B_1)_{\nu} \label{eq_chaincor_2_bounds} 
  \end{align}
    where the supremum and infimum range over density operators $\nu$ such that $\nu_{A_2 B_2 | A_1 B_1} = \rho_{A_2 B_2  | A_1 B_1}$ holds. 
\end{theorem}

\begin{proof}
We apply Theorem~\ref{thm_chainq_exact} with $B = B_1B_2$. The Markov chain condition implies that $H_\alpha(A_1|B_1B_2)_{\rho} = H_\alpha(A_1|B_1)_{\rho}$. To see this for $\alpha \in (\frac12, \infty)$, we could use the recoverability condition~\eqref{eq_recovery} for Markov chains together with the monotonicity of $D_{\alpha}$ under quantum channels~\cite{MDSFT13,WWY13,Beigi,FrankLieb}. We can also see it for all $\alpha \in (0, \infty)$ using the structure of a Markov chain stated in~\eqref{eq_Markovcondtrad}. Namely, there exists a decomposition $\bigoplus_{j} a_j \otimes b_j$ of the system $B_1$ such that
   \begin{align} \label{eq_rhoMcond2} 
  \rho_{A_1 B_1 B_2} & \cong \bigoplus_{j}  q_j \,  \hat{\rho}_{A_1 a_j} \otimes \hat{\rho}_{b_j B_2} 
  \end{align} 
  holds, where $\{q_j\}$ is a probability distribution and where $\{\hat{\rho}_{A a_j}\}$ and $\{\hat{\rho}_{b_j B}\}$ are families of density operators. Then, 
  \begin{align*}
  H_{\alpha}(A_1 | B_1 B_2)_{\rho} &= \frac{1}{1-\alpha} \log \tr\left( \bigl( \rho_{B_1 B_2}^{\frac{-\alpha'}{2}} \rho_{A_1 B_1 B_2} \rho_{B_1 B_2}^{\frac{-\alpha'}{2}} \bigr)^{\alpha} \right) \\
  &= \frac{1}{1-\alpha} \log  \tr\left( \bigoplus_{j} q_j \bigl( \hat{\rho}_{a_j}^{\frac{-\alpha'}{2}} \otimes \hat{\rho}_{b_j B_2}^{\frac{-\alpha'}{2}} \left( \hat{\rho}_{A_1 a_j} \otimes \hat{\rho}_{b_j B_2} \right)  \hat{\rho}_{a_j}^{\frac{-\alpha'}{2}} \otimes \hat{\rho}_{b_j B_2}^{\frac{-\alpha'}{2}} \bigr)^{\alpha} \right) \\
  &= \frac{1}{1-\alpha} \log  \tr\left( \bigoplus_{j} q_j \bigl( \hat{\rho}_{a_j}^{\frac{-\alpha'}{2}}   \hat{\rho}_{A_1 a_j}   \hat{\rho}_{a_j}^{\frac{-\alpha'}{2}} \bigr)^{\alpha} \otimes \hat{\rho}_{b_j B_2} \right) = H_{\alpha}(A_1 | B_1)_{\rho} \ .
  \end{align*}
 To prove~\eqref{eq_chaincor_2_bounds}, it only remains to show that the state $\nu_{A_1A_2B_1B_2}$ defined in~\eqref{eq_alphastate2} satisfies $\nu_{A_2B_2|A_1B_1} = \rho_{A_2B_2|A_1B_1}$. For that, we again use the fact that $\rho_{A_1B_1B_2}$ forms a Markov chain.  As we will be using this statement later in other contexts, we state it as a claim.
 \begin{claim}
 \label{claim_markov_cond}
 Let $\rho_{A_1 B_1 A_2 B_2}$ be a density operator such that  the Markov chain condition $A_1 \leftrightarrow B_1 \leftrightarrow B_2$ holds, let $\alpha \in (0, \infty)$ and let $\nu_{A_1 B_1 A_2 B_2}$ be as in~\eqref{eq_alphastate2} with $B \rightarrow B_1 B_2$.
 Then $\nu_{A_2 B_2 | A_1 B_1} = \rho_{A_2 B_2 | A_1 B_1}$.
 \end{claim}
 Letting $Z = \tr\left( \rho_{A_1B_1B_2}^{\frac{1}{2}} \rho^{-\alpha'}_{B_1B_2} \rho_{A_1B_1B_2}^{\frac{1}{2}}\right)^{\alpha}$,  the decomposition~\eqref{eq_rhoMcond2} allows us to write
\begin{align}
    \nu_{A_1B_1B_2} &= \frac{1}{Z} \left(\rho_{A_1B_1B_2}^{\frac{1}{2}} \rho_{B_1B_2}^{-\alpha'} \rho_{A_1B_1B_2}^{\frac{1}{2}} \right)^{\alpha} \\
&= \frac{1}{Z} \left(\bigoplus_{j} q_j^{1-\alpha'} \hat{\rho}_{A_1 a_j}^{\frac{1}{2}} \hat{\rho}_{a_j}^{-\alpha'} \hat{\rho}_{A_1 a_j}^{\frac{1}{2}} \otimes \hat{\rho}_{b_j B_2}^{1-\alpha'} \right)^{\alpha} \\
&= \frac{1}{Z} \bigoplus_{j} q_j \left(\hat{\rho}_{A_1 a_j}^{\frac{1}{2}} \hat{\rho}_{a_j}^{-\alpha'} \hat{\rho}_{A_1 a_j}^{\frac{1}{2}} \right)^{\alpha} \otimes \hat{\rho}_{b_j B_2} \ .
\end{align}  
It follows that 
\begin{align}
\nu_{A_1B_1}^{-\frac{1}{2}} \nu_{A_1B_1B_2}^{\frac{1}{2}} &= \bigoplus_{j} \hat{\rho}^0_{A_1 a_j} \otimes \hat{\rho}_{b_j}^{-\frac{1}{2}} \hat{\rho}_{b_j B_2}^{\frac{1}{2}}  \\
&= \rho_{A_1B_1}^{-\frac{1}{2}} \rho_{A_1B_1B_2}^{\frac{1}{2}} \ ,
\end{align}  
where $\hat{\rho}^0_{A_1 a_j}$ is the projector onto the support of $\hat{\rho}_{A_1 a_j}$. We used for the first equality the fact that the support of the operator $\left(\hat{\rho}_{A_1 a_j}^{\frac{1}{2}} \hat{\rho}_{a_j}^{-\alpha'} \hat{\rho}_{A_1 a_j}^{\frac{1}{2}} \right)^{\alpha}$ is the same as the support of $\hat{\rho}_{A_1 a_j}$.
As a result, we find
\begin{align}
\nu_{A_2B_2|A_1B_1} &=  \nu_{A_1B_1}^{-\frac{1}{2}} \nu_{A_1B_1B_2}^{\frac{1}{2}} \rho_{A_2|A_1B_1B_2} \nu_{A_1B_1B_2}^{\frac{1}{2}} \nu_{A_1B_1}^{-\frac{1}{2}} \\
&= \rho_{A_1B_1}^{-\frac{1}{2}} \rho_{A_1B_1B_2}^{\frac{1}{2}} \rho_{A_2|A_1B_1B_2} \rho_{A_1B_1B_2}^{\frac{1}{2}} \rho_{A_1B_1}^{-\frac{1}{2}} \\
&= \rho_{A_2B_2 | A_1B_1} \ .
\end{align}
This concludes the proof of Claim~\ref{claim_markov_cond} and gives the desired statement.
\end{proof}

The following simple corollary expresses the above chain rules in terms of quantum channels, i.e., trace preserving completely positive  (TPCP) maps, rather than conditional states.

\begin{corollary} \label{cor_conditionalmap}
Let $\rho^0_{R A_1 B_1}$ be a density operator on $R \otimes A_1 \otimes B_1$, $\cM = \cM_{A _2 B_2 \leftarrow R}$ be a TPCP map and $\alpha \in (0, \infty)$. Assuming that $\rho_{A_1 B_1 A_2 B_2} = \cM(\rho^0_{RA_1B_1})$ satisfies the Markov condition $A_1 \leftrightarrow B_1 \leftrightarrow B_2$, we have
   \begin{align*} 
          \inf_{\omega} H_{\alpha}(A_2 | B_2 A_1 B_1)_{\cM(\omega)} 
  \leqslant H_\alpha(A_1 A_2 | B_1 B_2)_{\cM(\rho^0)} - H_{\alpha}(A_1 | B_1)_{\rho^0} \leqslant \sup_{\omega} H_{\alpha}(A_2 | B_2 A_1 B_1)_{\cM(\omega)} 
 \end{align*}
    where the supremum and infimum range over density operators $\omega_{R A_1 B_1}$ on $R \otimes A_1 \otimes B_1$. Moreover, if $\rho^0_{R A_1 B_1}$ is pure then we can optimise over pure states $\omega_{R A_1 B_1}$.
\end{corollary}

\begin{proof} 
We apply Theorem~\ref{thm_chainq} to $\rho_{A_1 B_1 A_2 B_2}$. It suffices to show that the optimisation over $\nu_{A_1 B_1 A_2 B_2}$ satisfying $\nu_{A_2 B_2 | A_1 B_1} = \rho_{A_2 B_2 | A_1 B_1}$ is contained in the optimisation over $\omega_{RA_1B_1}$. For this, let  $\nu_{A_1 B_1 A_2 B_2}$ be any density operator satisfying $\nu_{A_2 B_2 | A_1 B_1} = \rho_{A_2 B_2 | A_1 B_1}$, i.e., 
  \begin{align} \label{eq_nurhorelation}
    \nu_{A_1 B_1 A_2 B_2} 
   =  \nu_{A_1 B_1}^{\frac{1}{2}} \rho_{A_2 B_2 | A_1 B_1} \nu_{A_1 B_1}^{\frac{1}{2}}  \ .
  \end{align}  
Now we choose 
  \begin{align*}
    \omega_{R A_1 B_1} = \nu_{A_1 B_1}^{\frac{1}{2}} \rho_{A_1 B_1}^{-\frac{1}{2}} \rho^0_{R A_1 B_1}  \rho_{A_1 B_1}^{-\frac{1}{2}}  \nu_{A_1 B_1}^{\frac{1}{2}} \ .
  \end{align*}
We then see that 
\begin{align*}
\cM(\omega_{R A_1 B_1}) &= \nu_{A_1 B_1}^{\frac{1}{2}} \rho_{A_1 B_1}^{-\frac{1}{2}} \cM(\rho^0_{R A_1 B_1})  \rho_{A_1 B_1}^{-\frac{1}{2}}  \nu_{A_1 B_1}^{\frac{1}{2}} \\
&= \nu_{A_1 B_1}^{\frac{1}{2}} \rho_{A_2 B_2 | A_1 B_1 }  \nu_{A_1 B_1}^{\frac{1}{2}} \\
&= \nu_{A_1 B_1 A_2 B_2} \ .
\end{align*}
\end{proof} 

\section{Entropy accumulation}\label{sec_accumulation}

This section is devoted to the main result on entropy accumulation.  The statement is formulated in  its fully general form as Theorem~\ref{thm:entropyaccumulationext} and presented in a slightly simplified version as Corollary~\ref{cor:EAT}.  We also give a formulation that corresponds to statement~\eqref{eq_entropyaccumulation} of the introduction (Corollary~\ref{cor_diffi}). Finally, we show how the Quantum Asymptotic Equipartition Property follows as a special case (cf.\ Corollary~\ref{cor_QAEP}). 

For $i \in \{1,\dots,n\}$, let $\cM_i$ be a TPCP map from $R_{i-1}$ to $X_i A_i B_i R_i$, where $A_i$ is finite-dimensional and where $X_i$ represents a classical value from an alphabet $\cX$ that is  determined by $A_i$ and $B_i$ together. More precisely, we require that, $\cM_{i} = \cT_{i} \circ \cM'_i$ where $\cM'_{i}$ is an arbitrary TPCP map from $R_{i-1}$ to $A_{i} B_{i} R_{i}$  and $\cT_i$ is a TPCP map from $A_{i}B_{i}$ to $X_{i} A_i B_i$ of the form
\begin{align}
\label{eq_form_extraction}
\cT_{i}(W_{A_iB_i}) = \sum_{y \in \cY , z \in \cZ} (\Pi_{A_i,y} \otimes \Pi_{B_i,z}) W_{A_i B_i} (\Pi_{A_i,y} \otimes \Pi_{B_i,z}) \otimes \proj{t(y,z)}_{X_i} \ ,
\end{align}
where $\{\Pi_{A_i, y}\}$ and $\{\Pi_{B_i, z}\}$ are families of mutually orthogonal projectors on $A_i$ and $B_i$, and where $t : \cY \times \cZ  \to \cX$ is a deterministic function (cf.\ Figs.~\ref{fig:rho-structure} and~\ref{fig:rho-structure-X}). Special cases of interest are when $X_i$ is trivial and $\cT_{i}$ is the identity map, and when  $X_i = t(Y_i, Z_i)$ where  $Y_i$ and $Z_i$  are classical parts of $A_i$ and $B_i$, respectively. Note that the maps $\cT_i$ have the property that, for any operator $\bar{W}_{X_iA_iB_i}$, if $\bar{W}_{X_iA_iB_i} = \cT_{i}(W_{A_i B_i})$ then $\bar{W}_{X_i A_i B_i} = \cT_{i}(\bar{W}_{A_i B_i})$.

The entropy accumulation theorem stated below will hold for states of the form 
\begin{align} \label{eq_rhomap}
  \rho_{A_1^n B_1^n X_1^n E} = ({\cM_n \circ \dots \circ \cM_1} \otimes \cI_E)(\rho^0_{R_0 E})
\end{align}
where $\rho^0_{R_0 E} \in  \mathrm{D}(R_0 \otimes E)$ is a density operator on $R_0$ and an arbitrary system $E$.  In addition, we require that the Markov conditions 
\begin{align} \label{eq_Markovgen}
  A_1^{i-1} \leftrightarrow B_1^{i-1}  E \leftrightarrow B_{i}
\end{align}
be satisfied for all $i \in \{1, \ldots, n\}$. 

Let $\mbP$ be the set of probability distributions on the alphabet $\cX$ of $X_i$, and let $R$ be a system isomorphic to $R_{i-1}$. For any $q \in \mbP$ we define the set of states 
\begin{align}
\label{eq_def_set_Sigma}
  \Sigma_i(q) = \bigl\{\nu_{X_i A_i B_i R_i R} = (\cM_i \otimes \cI_R)(\omega_{R_{i-1} R}) : \quad  \omega \in \mathrm{D}(R_{i-1} \otimes R) \text{ and } \nu_{X_i} = q \bigr\}  \ ,
\end{align}
where $\nu_{X_i}$ denotes the probability distribution over $\cX$ with the probabilities given by $\bra{x} \nu_{X_i} \ket{x}$. 
\begin{definition} \label{def_tradeoff}
  A  real function $f$ on $\mbP$ is called a \emph{min-} or \emph{max-tradeoff function} for $\cM_i$ if it  satisfies 
\begin{align*}
    f(q) \leqslant \inf_{\nu \in \Sigma_i(q)} H(A_i | B_i R)_{\nu} 
  \qquad \text{or} \qquad
    f(q) \geqslant \sup_{\nu \in \Sigma_i(q)} H(A_i |B_i R)_{\nu} \ ,
\end{align*}
respectively.\footnote{If the set $\Sigma_i(q)$ is empty then the infimum and supremum are by definition equal to~$\infty$ and $-\infty$, respectively, so that the conditions are trivial.}
\end{definition}

\begin{remark} \label{rem_Rdimension}
To determine the infimum $\inf_{\nu \in \Sigma_i(q)} H(A_i | B_i R)_{\nu}$, we may assume that $\omega_{R_{i-1} R}$ in the definition of $\Sigma_i(q)$ is pure. In fact, including a purifying system in $R$ cannot increase $H(A_i | B_i R)$ because of strong subadditivity. Similarly, to calculate the supremum $\sup_{\nu \in \Sigma_i(q)} H(A_i | B_i R)_{\nu}$, we may assume that $\omega_{R_{i-1} R}$ is a product state or that $R$ is trivial. This justifies the fact that we assumed $R$ is isomorphic to $R_{i-1}$ in the definition of $\Sigma_i(q)$. 
\end{remark}

\begin{remark} \label{rem_Rdimensiona}
As we will see in the proof below, one can also impose the constraint on the set  $\Sigma_i(q)$ that the system $R$ be isomorphic to $A_1^{i-1}B_1^{i-1}E$. Furthermore,  if a part of the latter is classical in $\rho$, one can restrict $\Sigma_i(q)$ to states satisfying this property.
\end{remark}

In the following, we denote by $\nabla f$ the gradient of a function $f$. (Note that in Theorem~\ref{thm:entropyaccumulationext} and Proposition~\ref{prop:accumulation-alpha} $f$ is an affine function, so that $\nabla f$ is a constant.) We write $\freq{X_1^n}$ for the distribution on $\cX$ defined by $\freq{X_1^n}(x) = \frac{|\{i \in \{1,\dots,n\} : X_i = x\}|}{n}$.  We also recall that in this context, an event $\Omega$ is defined by a subset of $\cX^n$ and we write $\rho[\Omega] = \sum_{x_1^n \in \Omega}\tr(\rho_{A_1^n B_1^n E, x_1^n})$ for the probability of the event $\Omega$ and 
\begin{align*}
  \rho_{X_1^n A_1^n B_1^n E | \Omega} = \frac{1}{\rho[\Omega]} \sum_{x_1^n \in \Omega} \proj{x_1^n} \otimes \rho_{A_1^n B_1^n E, x_1^n}
\end{align*}
for the state conditioned on~$\Omega$ (cf.\ Section~\ref{sec_Notation}).

\begin{theorem}\label{thm:entropyaccumulationext}
    Let $\cM_1,\dots,\cM_n$ and  $\rho_{A_1^n B_1^n X_1^n E}$ be such that~\eqref{eq_rhomap} and the Markov conditions~\eqref{eq_Markovgen} hold, let  $h \in \mathbb{R}$, let $f$ be an affine min-tradeoff function for $\cM_1,\dots,\cM_n$,  and let  $\varepsilon \in (0,1)$. Then, for any event $\Omega \subseteq \cX^n$ that implies $f(\freq{X_1^n}) \geqslant h$,\footnote{We say that the event $\Omega$ implies $f(\freq{X_1^n}) \geqslant h$ if for every $x_1^n \in \Omega, f(\freq{x_1^n}) \geqslant h$.}
      \begin{align}
        \label{eqn:eat-min}H_{\min}^{\varepsilon}(A_1^n | B_1^n E)_{\rho_{|\Omega}} & > n h -  c \sqrt{n}
  \end{align}
  holds for $c = 2 \bigl(\log (1+2 d_A) + \left\lceil \| \nabla f \|_\infty \right\rceil  \bigr) \sqrt{1- 2 \log (\varepsilon \rho[\Omega])}$, where $d_A$ is the maximum dimension of the systems~$A_i$.  Similarly,
  \begin{align} 
        \label{eqn:eat-max}H_{\max}^{\varepsilon}(A_1^n | B_1^n E)_{\rho_{|\Omega}} & < n h + c \sqrt{n} 
    \end{align}
    holds if $f$ is replaced by an affine max-tradeoff function and if $\Omega$ implies $f(\freq{X_1^n}) \leqslant h$. 
    \end{theorem}

Before proceeding to the proof, some remarks are in order. The first is that the Markov chain assumption on the state is important as argued in Appendix~\ref{app_Markov}. Secondly, the system $E$ could have been included in $B_1$, but for the applications we consider, it is clearer to keep a separate system $E$ that is not affected by the processes $\cM_1, \dots, \cM_n$. Thirdly, concerning the second order term, it is possible to replace $d_A$ with appropriate entropic quantities, as in the Quantum Asymptotic Equipartition Property~\cite{TCR09}, which could be useful when the systems $A_i$ are infinite-dimensional. The dependence of the second order term in the state and in the tradeoff function $f$ is studied in more detail in the subsequent work~\cite{DF18}. Finally, we note that the constraint that the tradeoff function be affine is not a severe restriction: given a convex min-tradeoff function, one can always choose a tangent hyperplane at a point of interest as an affine lower bound. This is illustrated in Corollary~\ref{cor:EATconvex}.

To prove the theorem, we will first show the following proposition, which is essentially a R\'enyi version of entropy accumulation. We then show how Theorem \ref{thm:entropyaccumulationext} follows from this proposition.

\begin{proposition}\label{prop:accumulation-alpha}
Let $\cM_1, \ldots, \cM_n$ and $\rho_{A_1^n B_1^n X_1^n E}$ be such that~\eqref{eq_rhomap} and the Markov conditions~\eqref{eq_Markovgen} hold, let  $h \in \mathbb{R}$, and let $f$ be an affine min-tradeoff function $f$ for $\cM_1,\dots,\cM_n$. Then, for any event $\Omega$ which implies  $f(\freq{X_1^n}) \geqslant h$,
      \begin{align}
          \label{eqn:eat-min-alpha} 
                  H^{\uparrow}_{\alpha}(A_1^n | B_1^n E)_{\rho_{|\Omega}} 
        & > n h - n \left( \frac{\alpha-1}{4} \right) V^2 - \frac{\alpha}{\alpha - 1} \log \frac{1}{\rho[\Omega]}
        \end{align}
      holds for $\alpha$ satisfying $1 < \alpha < 1 + \frac{2}{V}$, and $V = 2 \left\lceil \| \nabla f \|_\infty \right\rceil + 2  \log (1+2 d_A)$, where $d_A$ is the maximum dimension of the systems~$A_i$.  Similarly,
  \begin{align} 
      \label{eqn:eat-max-alpha}H_{\frac{1}{\alpha}}(A_1^n | B_1^n E)_{\rho_{|\Omega}} & < n h + n \left( \frac{\alpha-1}{4} \right)V^2 + \frac{\alpha}{\alpha - 1} \log \frac{1}{\rho[\Omega]}
    \end{align}
    holds if $f$ is replaced by an affine max-tradeoff function and if $\Omega$ implies $f(\freq{X_1^n}) \leqslant h$.
\end{proposition}

\begin{proof}
   We focus on proving the first inequality~\eqref{eqn:eat-min-alpha}. The proof of the second inequality~\eqref{eqn:eat-max-alpha} is similar, we only point out the main differences in the course of the proof. 
    
    The first step of the proof is to construct a state that will allow us to lower-bound $H^\uparrow_{\alpha}(A_1^n | B_1^n E)_{\rho_{|\Omega}}$ using the chain rule of Theorem~\ref{thm_chainq}, while ensuring that the tradeoff function is taken into account.  Let $[g_{\min}, g_{\max}]$ be the smallest real interval that contains the range $f(\mbP)$ of $f$, and set $\bar{g} = \frac{1}{2} (g_{\min} + g_{\max})$. Furthermore, for every $i$, let $\mathcal{D}_i : X_i \rightarrow X_i D_i \bar{D}_i$, with $\dim D_i = \dim \bar{D}_i$, be a TPCP map defined as
\begin{align*}
    \mathcal{D}_i(W_{X_i}) = \sum_{x \in \mathcal{X}} \bra{x}W_{X_i} \ket{x} \cdot \proj{x}_{X_i} \otimes \tau(x)_{D_i \bar{D}_i} \ ,
\end{align*}
    where $\tau(x)$ is a mixture between a maximally entangled state on $D_i \otimes \bar{D}_i$ and a fully mixed state such that the marginal on $\bar{D}_i$ is uniform and such that $H_{\alpha}(D_i|\bar{D}_i)_{\tau(x)} = \bar{g} - f(\delta_x)$ (here $\delta_x$ stands for the distribution with all the weight on element $x$). To ensure that this is possible, we need to choose $\dim D_i$ large enough, so we need to bound how large $\bar{g} - f(\delta_x)$ can be, positive or negative.  By the definition of $\bar{g}$,  $|\bar{g} - f(\delta_x)|$ cannot be larger than $\frac{1}{2} |g_{\max} - g_{\min}|  \leqslant  \| \nabla f \|_\infty$.  We therefore take the dimension of the spaces $D_i$ to be equal to  
    \begin{align*}
     d_D := \left\lceil  2^{\| \nabla f \|_\infty} \right\rceil  \leqslant  2^{\left\lceil \| \nabla f \|_\infty \right\rceil } \ .
     \end{align*}
    For later use, we note that we have
    \begin{align} \label{eq_dADboundext}
     \log(1 + 2 d_{A} d_D) \leqslant \left\lceil \| \nabla f \|_\infty  \right\rceil + \log (1+2 d_A ) = \frac{V}{2}.
    \end{align}
    Now, let
    \begin{align}
    \label{eq_def_rho_bar}
      \bar{\rho} := (\mathcal{D}_n \circ \dots \circ \mathcal{D}_1)(\rho) \ . 
    \end{align}
  Note that $\bar{\rho}_{X_1^n A_1^n B_1^n E} = \rho_{X_1^n A_1^n B_1^n E}$.
    
  One can think of the $D$ systems as an ``entropy price'' that encodes the tradeoff function. With these systems in place, the output entropy includes an extra term that allows the tradeoff function to be taken into account in the optimisation arising in Theorem~\ref{thm_chainq}. This is formalised by the following facts, which are proven in Claim \ref{claim:proof-of-main-thm}:
   \begin{align}
    H^{\uparrow}_{\alpha}(A_1^n | B_1^n E)_{\rho_{|\Omega}}
        &\geqslant   H^{\uparrow}_{\alpha}(A_1^n D_1^n | B_1^n E \bar{D}_1^n)_{\bar{\rho}_{|\Omega}} - n \bar{g} + n h \ , \\ 
    H_{\frac{1}{\alpha}}(A_1^n | B_1^n E)_{\rho_{|\Omega}}
        &\leqslant   H_{\frac{1}{\alpha}}(A_1^n D_1^n | B_1^n E \bar{D}_1^n)_{\bar{\rho}_{|\Omega}} - n \bar{g} + n h \ . 
    \end{align}
  The next step is to relate the entropies on the conditional state $\rho_{|\Omega}$ to those on the unconditional state. To do this, we use Lemmas~\ref{lem:abx-chain-rule-opt} and~\ref{lem:abx-chain-rule} applied to $\bar{\rho} = \rho[\Omega] \bar{\rho}_{|\Omega} + (\bar{\rho} - \rho[\Omega] \bar{\rho}_{|\Omega})$, together with the fact that $H_{\alpha}^{\uparrow} \geqslant H_{\alpha}$, and obtain\footnote{Note that the reason why we use $H_{\alpha}^{\uparrow}$ in one case but not the other is due to the difference between Lemmas~\ref{lem:abx-chain-rule-opt} and~\ref{lem:abx-chain-rule}.}
   \begin{align}
   \label{eq_lb_using_d}
    H^{\uparrow}_{\alpha}(A_1^n | B_1^n E)_{\rho_{|\Omega}}
          &\geqslant   H_{\alpha}(A_1^n D_1^n | B_1^n E \bar{D}_1^n)_{\bar{\rho}} - \frac{\alpha}{\alpha - 1} \log\frac{1}{\rho[\Omega]}  - n \bar{g} + n h \ , \\           \label{eq_ub_using_d}
    H_{\frac{1}{\alpha}}(A_1^n | B_1^n E)_{\rho_{|\Omega}}
          &\leqslant   H_{\frac{1}{\alpha}}(A_1^n D_1^n | B_1^n E \bar{D}_1^n)_{\bar{\rho}} + \frac{\alpha}{\alpha - 1} \log\frac{1}{\rho[\Omega]}  - n \bar{g} + n h \ . 
  \end{align}
   
To show the desired inequality~\eqref{eqn:eat-min-alpha}, it now suffices to prove that $H_{\alpha}(A_1^n D_1^n | B_1^n E \bar{D}_1^n)_{\bar{\rho}}$ is lower bounded by (roughly) $n \bar{g}$. To do that, we are now going to use the chain rule for R\'enyi entropies in the form of Corollary~\ref{cor_conditionalmap} $n$ times on the state $\bar{\rho}$, with the following substitutions at step $i$:
    \begin{itemize}
        \item $A_1 \rightarrow A_1^{i-1} D_1^{i-1}$
        \item $B_1 \rightarrow B_1^{i-1} E \bar{D}_1^{i-1} $
        \item $A_2 \rightarrow A_i D_i$
        \item $B_2 \rightarrow B_i \bar{D}_i$
        \item $R \rightarrow R_{i-1}$
        \item $\cM \rightarrow \tr_{X_i} \circ \cD_i \circ \cM_i$.
    \end{itemize}
    To establish the Markov chain condition, we compute the conditional mutual information. Using the chain rule, we obtain
    \begin{align}
    \notag
    &I(A_1^{i-1} D_1^{i-1} : B_i \bar{D}_i | B_1^{i-1} E \bar{D}_1^{i-1}) \\
    &= I(A_1^{i-1}  : B_i \bar{D}_i | B_1^{i-1} E \bar{D}_1^{i-1}) + I(D_1^{i-1}  : B_i \bar{D}_i | A_1^{i-1} B_1^{i-1} E \bar{D}_1^{i-1}) .
    \label{eq:cmi_check_markov_cond}
    \end{align}
    We first show that the second term is zero. By construction, $D_1^{i-1} \bar{D}_1^{i-1}$ conditioned on $X_1^{i-1}$ is independent of all the other systems.  This implies that 
    $I(D_1^{i-1} \bar{D}_1^{i-1}  : B_i \bar{D}_i | X_1^{i-1} A_1^{i-1} B_1^{i-1} E ) = 0$. In addition, using the fact that $X_1^{i-1}$ is determined by $A_1^{i-1} B_1^{i-1}$, the systems $X_1^{i-1}$ can be removed from the conditioned without changing the value. Then, using the chain rule and together with the non-negativity of the conditional mutual information, this shows that $I(D_1^{i-1}  : B_i \bar{D}_i | A_1^{i-1} B_1^{i-1} E \bar{D}_1^{i-1}) = 0$. To compute the first term in~\eqref{eq:cmi_check_markov_cond}, we use the fact that $\bar{D}_1^n$ is uniform independently of $A_1^n B_1^n E$ so that $I(A_1^{i-1}  : B_i \bar{D}_i | B_1^{i-1} E \bar{D}_1^{i-1}) = I(A_1^{i-1}  : B_i | B_1^{i-1} E )$. But then the assumed Markov condition on $\rho_{A_1^n B_1^n E}$ implies that this quantity is zero and establishes the required condition to apply Corollary~\ref{cor_conditionalmap}.
 We thus obtain
    \begin{align}
    \notag
    H_{\alpha}(A_1^n D_1^n | B_1^n E \bar{D}_1^n)_{\bar{\rho}}
         &\geqslant \sum_i \inf_{\omega_{R_{i-1} R}} H_{\alpha}(A_i D_i|B_i \bar{D}_i R)_{(\mathcal{D}_i \circ \cM_i)(\omega)}\\
         \notag
         & > \sum_i \inf_{\omega_{R_{i-1} R}} H(A_i D_i|B_i \bar{D}_i R)_{(\mathcal{D}_i \circ \cM_i)(\omega)} - n(\alpha-1) \log^2(1 + 2 d_{A} d_D)  \\
         & \geqslant \sum_i \inf_{\omega_{R_{i-1} R}} H(A_i D_i|B_i \bar{D}_i R)_{(\mathcal{D}_i \circ \cM_i)(\omega)} - n \frac{(\alpha-1)}{4} V^2 \ ,
         \label{eq_bound_sum_inf_h}
    \end{align}
    where we have invoked Lemma \ref{lem_HalphaH} in the second inequality and~\eqref{eq_dADboundext} in the last. Note that the restriction of this lemma  that $\alpha$ satisfy $1 < \alpha <  1+ 1/ \log(1 + 2 d_A d_D)$  is implied by our assumption that $\alpha < 1 + 2/V$.  The infimum is taken over all states $\omega_{R_{i-1} R}$, where the system $R$ is isomorphic to $A_1^{i-1} D_1^{i-1} B_1^{i-1} \bar{D}_1^{i-1} E$. This condition can be further strengthened by redoing the above argument with Theorem~\ref{thm_chainq_exact} instead of Corollary~\ref{cor_conditionalmap}. It turns out that the system $R$ can be taken to be isomorphic to $A_1^{i-1} B_1^{i-1} E$, as noted in Remark~\ref{rem_Rdimensiona}.\footnote{The full proof of this fact is available in the source code of this file on the arXiv. To access it, follow the instructions on the line labelled ``EXTRA'' in the preamble.}

    Considering the right hand side of expression~\eqref{eq_bound_sum_inf_h}, we get for any such state $\omega_{R_{i-1} R}$, 
      \begin{align*}
        H(A_i D_i | B_i \bar{D}_i R)_{(\mathcal{D}_i \circ \cM_i)(\omega)}
        & = H(A_i X_i D_i | B_i \bar{D}_i R)_{(\mathcal{D}_i \circ \cM_i)(\omega)} \\
        &= H(A_i X_i | B_i R)_{\cM_i(\omega)}  + H(D_i|\bar{D}_i X_i)_{(\mathcal{D}_i \circ \cM_i)(\omega)}\\
        &= H(A_i  | B_i R)_{\cM_i(\omega)}  + \sum_x q(x) H(D_i|\bar{D}_i)_{\tau(x)}\\
         & \geqslant H(A_i | B_i R)_{\cM_i(\omega)}  + \sum_x q(x) H_{\alpha}(D_i|\bar{D}_i )_{\tau(x)}\\
         & = H(A_i | B_i R)_{\cM_i(\omega)}  + \sum_x q(x) \bigl( \bar{g} - f(\delta_x) \bigr) \\
         &= H(A_i | B_i R)_{\cM_i(\omega)} + \bar{g} - f(q)\\
         &\geqslant \bar{g} \ 
        \end{align*}
where $q = \cM_i(\omega)_{X_i}$ denotes the distribution of $X_i$ on $\mathcal{X}$ obtained from the state $\cM_i(\omega)$. The third equality comes from the fact that $X_i$ is determined by $A_iB_i$. The first inequality follows from the monotonicity of the R\'enyi entropies in $\alpha$~\cite{Beigi,MDSFT13}. The last equality holds because $f$ is affine and the final inequality because $f$ is a min-tradeoff function. Putting everything together, Eq.~\eqref{eq_lb_using_d} becomes
   \begin{align*}
        H^{\uparrow}_{\alpha}(A_1^n | B_1^n E)_{\rho_{|\Omega}} > n h - n \frac{(\alpha-1)}{4} V^2 - \frac{\alpha-1}{\alpha} \log\frac{1}{\rho[\Omega]} \ .
   \end{align*}  
    This concludes the proof of the first inequality~\eqref{eqn:eat-min-alpha} of Proposition~\ref{prop:accumulation-alpha}. 
    
    In order to show the second inequality~\eqref{eqn:eat-max-alpha}, using the same argument as before, we obtain
      \begin{align*}
    H_{\frac{1}{\alpha}}(A_1^n D_1^n | B_1^n E \bar{D}_1^n)_{\bar{\rho}}
         & < \sum_i \sup_{\omega_{R_{i-1} R}} H(A_i D_i|B_i \bar{D}_i R)_{(\mathcal{D}_i \circ \cM_i)(\omega)} + n \frac{(\alpha-1)}{4} V^2 \ ,
    \end{align*}
where the supremum is over all states $\omega_{R_{i-1}R}$ with $R$ constrained as described by Remark~\ref{rem_Rdimensiona}. For any such state and a max-tradeoff function $f$, we have 
    \begin{align*}
        H(A_i D_i | B_i \bar{D}_i R)_{(\mathcal{D}_i \circ \cM_i)(\omega)}
         & \leqslant H(A_i | B_i R)_{\cM_i(\omega)}  + \sum_x q(x) H_{\frac{1}{\alpha}}(D_i|\bar{D}_i )_{\tau(x)}\\
         & = H(A_i | B_i R)_{\cM_i(\omega)}  + \sum_x q(x) \bigl( \bar{g} - f(\delta_x) \bigr) \\
         &= H(A_i | B_i R)_{\cM_i(\omega)} + \bar{g} - f(q)\\
         &\leqslant \bar{g} \ .
        \end{align*}
 It then suffices to combine these inequalities with inequality~\eqref{eq_ub_using_d}.
\end{proof}

We now prove the claim used in the preceding proof.

    \begin{claim}\label{claim:proof-of-main-thm}
   For $\alpha \in (1,2]$, $\rho$ and $\Omega$ as in the statement of Proposition~\ref{prop:accumulation-alpha} and $\bar{\rho}$ as defined in~\eqref{eq_def_rho_bar} (see also the preceding text for a definition of $\bar{g}$), we have
   \begin{align}
   \label{eq_lb_using_d-2}
    H^{\uparrow}_{\alpha}(A_1^n | B_1^n E)_{\rho_{|\Omega}}
    &\geqslant   H^{\uparrow}_{\alpha}(A_1^n D_1^n | B_1^n E \bar{D}_1^n)_{\bar{\rho}_{|\Omega}} - n \bar{g} + n h \ , \\ 
          \label{eq_ub_using_d-2}
    H_{\frac{1}{\alpha}}(A_1^n | B_1^n E)_{\rho_{|\Omega}}
        &\leqslant   H_{\frac{1}{\alpha}}(A_1^n D_1^n | B_1^n E \bar{D}_1^n)_{\bar{\rho}_{|\Omega}} - n \bar{g} + n h \ . 
          \end{align}
   \end{claim}
   \begin{proof}
   We focus on proving inequality \eqref{eq_lb_using_d-2}.
   The first step is to show that as $X_1^n$ is a deterministic function of $A_1^n B_1^n$, we have
       \begin{align}  \label{eq_Xdrop}
        H^{\uparrow}_{\alpha}(A_1^n D_1^n | B_1^n E \bar{D}_1^n)_{\bar{\rho}_{|\Omega}} &= H^{\uparrow}_{\alpha}(A_1^n X_1^n D_1^n | B_1^n E \bar{D}_1^n)_{\bar{\rho}_{|\Omega}} \ .
    \end{align}  
    In order to do that, observe that for any $x_1^n \in \cX^n$, we have
    \begin{align*}
    \bar{\rho}_{A_1^n B_1^n E D_1^n \bar{D}_1^n, x_1^n} = \rho_{A_1^n B_1^n E, x_1^n} \otimes \tau(x_1^n)_{D_1^n \bar{D}_1^n} \ ,
    \end{align*}
   where we introduced the notation $\tau(x_1^n)_{D_1^n \bar{D}_1^n} = \tau(x_1)_{D_1 \bar{D}_1} \otimes \cdots \otimes \tau(x_n)_{D_n \bar{D}_n}$.
    This implies that for any $x_1^n$, we have
         \begin{align*}
    \bar{\rho}_{X_1^n A_1^n B_1^n E D_1^n \bar{D}_1^n, x_1^n} = (\cT_n \circ \dots \circ \cT_1)(\bar{\rho}_{A_1^n B_1^n E D_1^n \bar{D}_1^n, x_1^n}) \ .
    \end{align*}
    By taking the sum over $x_1^n \in \Omega$ and then normalising by $\rho[\Omega]$, we get 
     \begin{align*}
    \bar{\rho}_{X_1^n A_1^n B_1^n E D_1^n \bar{D}_1^n | \Omega} = (\cT_n \circ \dots \circ \cT_1)(\bar{\rho}_{A_1^n B_1^n E D_1^n \bar{D}_1^n | \Omega}) \ .
    \end{align*}
Thus, we can apply Lemma~\ref{lem_det_func} and prove the equality~\eqref{eq_Xdrop}.

    Let now $\sigma_{B_1^n E \bar{D}_1^n}$ be a state such that
    \begin{align*}
       H^\uparrow_{\alpha}(A_1^n X_1^n D_1^n | B_1^n E \bar{D}_1^n)_{\bar{\rho}_{|\Omega}} 
     = - D_{\alpha}(\bar{\rho}_{A_1^n X_1^n D_1^n B_1^n E \bar{D}_1^n | \Omega} \| \id_{A_1^n X_1^n D_1^n} \otimes \sigma_{B_1^n E \bar{D}_1^n}) \ .
    \end{align*}
    Let furthermore $\cS = \cS_{D \bar{D}}$ be the TPCP map that applies a random (according to the Haar measure) unitary to $D$ and its conjugate to $\bar{D}$ (in such a way that the maximally entangled state on $D \bar{D}$ used to define $\tau(x)$ is preserved). It is then easy to see that the map $\cS^{\otimes n}$ applied to the $n$ pairs $D_i \bar{D}_i$ leaves $\bar{\rho}_{|\Omega}$ invariant. Hence, by the data processing inequality
    \begin{align*}
      H^\uparrow_{\alpha}(A_1^n X_1^n D_1^n | B_1^n E \bar{D}_1^n)_{\bar{\rho}_{|\Omega}} 
    & \leqslant - D_{\alpha}(\cS^{\otimes n}(\bar{\rho}_{A_1^n X_1^n D_1^n B_1^n E \bar{D}_1^n | \Omega}) \| \cS^{\otimes n}(\id_{A_1^n X_1^n D_1^n} \otimes \sigma_{B_1^n E \bar{D}_1^n})) \\
    & =  - D_{\alpha}(\bar{\rho}_{A_1^n X_1^n D_1^n B_1^n E \bar{D}_1^n | \Omega} \| \id_{A_1^n X_1^n D_1^n} \otimes \bar{\sigma}_{B_1^n E \bar{D}_1^n})  \ ,
    \end{align*}
    where $\bar{\sigma}_{B_1^n E \bar{D}_1^n} = \sigma_{B_1^n E} \otimes \bar{\rho}_{\bar{D}_1^n}$. Lemma~\ref{lem_chainprep} then implies that
    \begin{align}   \label{eq_HalphaADdec}  
      H^\uparrow_{\alpha}(A_1^n X_1^n D_1^n | B_1^n E \bar{D}_1^n)_{\bar{\rho}_{|\Omega}} 
     \leqslant H^\uparrow_{\alpha}(A_1^n X_1^n | B_1^n E \bar{D}_1^n)_{\bar{\rho}_{|\Omega}} + H_{\alpha}(D_1^n | A_1^n X_1^n B_1^n E \bar{D}_1^n)_{\nu} \ . 
    \end{align}
    where $\nu$ is a state defined by 
    \begin{align*}
     \nu_{A_1^n X_1^n B_1^n E \bar{D}_1^n} &= \frac{\left(\bar{\rho}_{A_1^n X_1^n B_1^n E \bar{D}_{1}^n|\Omega}^{\frac12} \bar{\sigma}^{-\alpha'}_{B_1^n E \bar{D}_{1}^n} \bar{\rho}_{A_1^n X_1^n B_1^n E \bar{D}_{1}^n|\Omega}^{\frac12}\right)^{\alpha}}{\tr\left(\bar{\rho}_{A_1^n X_1^n B_1^n E \bar{D}_{1}^n|\Omega}^{\frac12} \bar{\sigma}^{-\alpha'}_{B_1^n E \bar{D}_{1}^n} \bar{\rho}_{A_1^n X_1^n B_1^n E \bar{D}_{1}^n|\Omega}^{\frac12}\right)^{\alpha}}  \quad \text{and} \\
     \nu_{A_1^n X_1^n B_1^n E D_1^n \bar{D}_1^n} &= \nu_{A_1^n X_1^n B_1^n E \bar{D}_1^n}^{\frac12} \bar{\rho}_{D_1^n | A_1^n X_1^n B_1^n E \bar{D}_1^n | \Omega} \nu_{A_1^n X_1^n B_1^n E \bar{D}_1^n}^{\frac12} \ .
    \end{align*}
    We now use properties of $\rho_{|\Omega}$ and $\bar{\sigma}$ to simplify the expression of $\nu$. Observing that
    \begin{align}
    \label{eq:Dbar_product}
    \bar{\rho}_{X_1^n A_1^n B_1^n E \bar{D}_1^n | \Omega} =\frac{1}{\rho[\Omega]} \sum_{x_1^n \in \Omega} \proj{x_1^n} \otimes \rho_{A_1 B_1^n E, x_1^n} \otimes \bar{\rho}_{\bar{D}_1^n} \ ,
    \end{align}
     we can write
        \begin{align*}
     \nu_{A_1^n X_1^n B_1^n E \bar{D}_1^n} 
     & = \sum_{x_1^n \in \Omega}  \proj{x_1^n} \otimes \nu_{A_1^n B_1^n E, x_1^n} \otimes \bar{\rho}_{\bar{D}_1^n} \ , \\
     & \text{ with} \quad \nu_{A_1^nB_1^n E, x_1^n} = \frac{1}{\rho[\Omega]^{\alpha}}\frac{ \left(  \rho_{A_1^n B_1^n E, x_1^n}^{\frac12} \sigma^{\frac{1-\alpha}{\alpha}}_{B_1^n E} \rho_{A_1^n B_1^n E,x_1^n}^{\frac12}\right)^{\alpha}}{\tr\left(\bar{\rho}_{A_1^n X_1^n B_1^n E|\Omega}^{\frac12} \sigma^{\frac{1-\alpha}{\alpha}}_{B_1^n E} \bar{\rho}_{A_1^n X_1^n B_1^n E|\Omega}^{\frac12}\right)^{\alpha}} \ .
      \end{align*}     
      In addition, as $\bar{\rho}_{|\Omega}$ is of the form 
      \begin{align*}
         \bar{\rho}_{A_1^n X_1^n B_1^n E D_1^n \bar{D}_1^n|\Omega} = \frac{1}{\rho[\Omega]} \sum_{x_1^n \in \Omega} \proj{x_1^n}_{X_1^n} \otimes  \rho_{A_1^n B_1^n E, x_1^n} \otimes \tau(x_1^n)_{D_1^n\bar{D}_1^n} \ ,
     \end{align*}
 we have
      \begin{align*}
        \bar{\rho}_{D_1^n | A_1^n X_1^n B_1^n E \bar{D}_1^n | \Omega}
        = \sum_{x_1^n \in \Omega} \proj{x_1^n}_{X_1^n} \otimes  \rho_{A_1^n B_1^n E, x_1^n}^{0} \otimes \tau(x_1^n)_{D_1^n | \bar{D}_1^n} \ ,
      \end{align*}
      where $\rho_{A_1^n B_1^n E, x_1^n}^{0}$ is the projector onto the support of $\rho_{A_1^n B_1^n E, x_1^n}$. Hence,
      \begin{align}
      \label{eq_expr_nu}
       \nu_{A_1^n X_1^n B_1^n E D_1^n \bar{D}_1^n}  = \sum_{x_1^n \in \Omega} \proj{x_1^n} \otimes  \nu_{A_1^n B_1^n E, x_1^n} \otimes \tau(x_1^n)_{D_1^n \bar{D}_1^n} \ .
      \end{align}

Getting back to the inequality~\eqref{eq_HalphaADdec}, we have $H^\uparrow_{\alpha}(A_1^n X_1^n | B_1^n E \bar{D}_1^n)_{\bar{\rho}_{|\Omega}} = H^\uparrow_{\alpha}(A_1^n | B_1^n E )_{\bar{\rho}_{|\Omega}}$ using Eq.~\eqref{eq:Dbar_product} to drop $\bar{D}_1^n$ and Lemma~\ref{lem_det_func} to drop $X_1^n$. Moreover, using~\eqref{eq_expr_nu}, we have that $H_{\alpha}(D_1^n | A_1^n X_1^n B_1^n E \bar{D}_1^n)_{\nu} = H_{\alpha}(D_1^n | X_1^n \bar{D}_1^n)_{\nu}$.
      Finally, we get
      \begin{align}\label{eq_Hsecondredext}
        H^{\uparrow}_{\alpha}(A_1^n D_1^n | B_1^n E \bar{D}_1^n)_{\bar{\rho}_{|\Omega}}
        \leqslant H^{\uparrow}_{\alpha}(A_1^n | B_1^n E)_{\rho_{|\Omega}} + H_{\alpha}(D_1^n | \bar{D}_1^n X_1^n)_{\nu} \ .
      \end{align}
It is a direct consequence of the definition of $\tau(x)$  that
    \begin{multline*}
        H_{\alpha}(D_1^n | \bar{D}_1^n)_{\tau(x_1^n)} 
        = n \bar{g} - \sum_{i=1}^n f(\delta_{x_i}) \\
        = n \bar{g} - n \sum_{x \in \cX} \freq{x_1^n}(x) f(\delta_x)
        = n \bar{g} - n f\left( \sum_{x \in \cX} \freq{x_1^n}(x) \delta_x\right) 
       = n \bar{g} - n f(\freq{x_1^n}) \ ,
    \end{multline*}
    where we have used that $f$ is an affine function. Using Lemma~\ref{lem_classicalsideinformation} and~\eqref{eq_expr_nu} we can bound the second term on the right hand side of~\eqref{eq_Hsecondredext} by 
    \begin{align*}
      H_{\alpha}(D_1^n | \bar{D}_1^n X_1^n)_{\nu} &\leqslant  \max_{x_1^n \in \Omega}    H_{\alpha}(D_1^n | \bar{D}_1^n)_{\tau({x_1^n})}\\ 
      &\leqslant \max_{x_1^n: \, f(\freq{x_1^n}) \geqslant h}     n \bar{g} - n f(\freq{x_1^n})    \leqslant n \bar{g} - n h \ .
    \end{align*}
    Inserting this in~\eqref{eq_Hsecondredext} gives
    \begin{align*}
        H^{\uparrow}_{\alpha}(A_1^n  | B_1^n E)_{\rho_{|\Omega}} \geqslant H^{\uparrow}_{\alpha}(A_1^n D_1^n | B_1^n E \bar{D}_1^n)_{\bar{\rho}_{|\Omega}} - n \bar{g} + n h \ .
    \end{align*}
    This concludes the proof of inequality~\eqref{eq_lb_using_d-2}. For the proof of inequality~\eqref{eq_ub_using_d-2}, we can follow similar steps.\footnote{The full proof of this case is available in the source code of this file on the arXiv. To access it, follow the instructions on the line labelled ``EXTRA'' in the preamble.}

\end{proof}

Finally, we prove Theorem \ref{thm:entropyaccumulationext} using Proposition~\ref{prop:accumulation-alpha}.
\begin{proof}[Proof of Theorem \ref{thm:entropyaccumulationext}]
    The first step is to use Lemma \ref{lem_Hepsalpha} to lower-bound the smooth min-entropy by a R\'enyi entropy:
    \begin{align}\label{eq_eatblock}
        H_{\min}^{\varepsilon}(A_1^n | B_1^n E)_{\rho_{|\Omega}} \geqslant H^{\uparrow}_{\alpha}(A_1^n | B_1^n E)_{\rho_{|\Omega}} - \frac{g(\varepsilon)}{\alpha-1} \ .
    \end{align}
    Then Proposition \ref{prop:accumulation-alpha} yields
    \begin{align*}
    H_{\min}^{\varepsilon}(A_1^n | B_1^n E)_{\rho_{|\Omega}} &> n h - n \frac{(\alpha-1)}{4} V^2 - \frac{1}{\alpha'} \log \frac{1}{\rho[\Omega]} - \frac{g(\varepsilon)}{\alpha-1}\\
            &> n h - n \frac{(\alpha-1)}{4} V^2 - \frac{1}{\alpha'} \log \frac{1}{\rho[\Omega]}  - \frac{\log (2/\varepsilon^2)}{\alpha-1} \\
            &\geqslant n h - n \frac{(\alpha-1)}{4} V^2 - \frac{1}{(\alpha-1)} \log \frac{2}{\rho[\Omega]^2 \varepsilon^2} \ , 
    \end{align*}
    where we have used the fact that we are constrained to choose $\alpha \leqslant 1 + \frac{2}{V} \leqslant 2$ in the last inequality. We now choose
    \begin{align} \label{eq_alphachoiceext}
       \alpha := 1 + \frac{2 \sqrt{\log \frac{2}{\rho[\Omega]^2 \varepsilon^2} }}{\sqrt{n} V} \ .
    \end{align}
    and note that, as long as
    \begin{align} \label{eq_nboundext}
     n >  \log \frac{2}{\rho[\Omega]^2 \varepsilon^2} \ ,
     \end{align}
     the value $\alpha$ is strictly smaller than $1 + \frac{2}{V}$ and therefore within the required bounds. Note also that if~\eqref{eq_nboundext} does not hold then the term $c \sqrt{n}$ in the claim~\eqref{eqn:eat-min} is at least $n V \geqslant 2 n \log(1+ 2 d_A) \geqslant 2 n \log d_A$, whereas the min-entropy is always at least $- n \log d_A$ and $n f_{\min}(q)$ is at most $n \log d_A$, which means that the claim is trivial. Finally, inserting~\eqref{eq_alphachoiceext} into the above yields
    \[ H_{\min}^{\varepsilon}(A_1^n | B_1^n E)_{\rho} > n h - \sqrt{n} V \sqrt{\log \frac{2}{\rho[\Omega]^2 \varepsilon^2}} \ , \]
    as advertised. Once again, the max-entropy statement (\ref{eqn:eat-max}) holds by switching the direction of the inequalities, flipping the appropriate signs, and replacing every occurrence of~$H^{\uparrow}_{\alpha}$ by~$H_{\frac{1}{\alpha}}$.
\end{proof}

It might seem restrictive to assume that the tradeoff function is affine. We next show that we may take a general convex function provided the event $\Omega$ can be described as follows: $x^n \in \Omega$ if and only if $\freq{x^n} \in \hat{\Omega}$ where $\hat{\Omega}$ is a convex subset of $\mbP$.
\begin{corollary}\label{cor:EATconvex}
    Let $\cM_1,\dots,\cM_n$ and  $\rho_{A_1^n B_1^n X_1^n E}$ be such that~\eqref{eq_rhomap} and the Markov conditions~\eqref{eq_Markovgen} hold, let  $h \in \mathbb{R}, \varepsilon \in (0,1)$, let $\hat{\Omega}$ be a convex set $\hat{\Omega} \subseteq \mbP$ and define the corresponding event $\Omega \subseteq \cX^n$ by $x_1^n \in \Omega \Leftrightarrow \freq{x_1^n} \in \hat{\Omega}$. Then, if $f$ is a differentiable and convex min-tradeoff function for $\cM_1,\dots,\cM_n$ satisfying $f(q) \geqslant h$ for all $q \in \hat{\Omega}$, we have
      \begin{align}
        \label{eqn:eat-min-convex}H_{\min}^{\varepsilon}(A_1^n | B_1^n E)_{\rho_{|\Omega}} & > n h -  c \sqrt{n}
  \end{align}
  where $c = 2 \bigl(\log (1+2 d_A) + \left\lceil \| \nabla f \|_\infty \right\rceil  \bigr) \sqrt{1- 2 \log (\varepsilon \rho[\Omega])}$. Similarly, if $f$ is a differentiable and concave max-tradeoff function for $\cM_1,\dots,\cM_n$ satisfying $f(q) \leqslant h$ for all $q \in \hat{\Omega}$, we have
   \begin{align} 
        \label{eqn:eat-max-convex}H_{\max}^{\varepsilon}(A_1^n | B_1^n E)_{\rho_{|\Omega}} & < n h + c \sqrt{n} \ .
    \end{align}
   \end{corollary}

\begin{proof}
Let us denote by $\mathrm{cl}(\hat{\Omega})$ the closure of the set $\hat{\Omega}$.
Now as $f$ is continuous on the compact set $\mathrm{cl}(\hat{\Omega})$ (it is even assumed to be differentiable on all of $\mbP$), we have $\min_{q \in \mathrm{cl}(\hat{\Omega})} f(q) = f(q_0)$ for some $q_0 \in \mathrm{cl}(\hat{\Omega})$. By continuity of $f$ and by definition of $h$, we have $f(q_0) \geq h$. Now consider the affine function $g(q) = (\nabla f)_{q_0} \cdot (q - q_0) + f(q_0)$. By convexity of $f$, we have that $g(q) \leq f(q)$ for all $q \in \mbP$ and thus $g$ is a min-tradeoff function. In addition, as $\mathrm{cl}(\hat{\Omega})$ is convex we can apply the first order optimality conditions and get that $(\nabla f)_{q_0} \cdot (q - q_0) \geq 0$ for all $q \in \mathrm{cl}(\hat{\Omega})$. As a result, for all $q \in \mathrm{cl}(\hat{\Omega})$, we have $g(q) \geqslant f(q_0) \geq h$. This implies that if $x_1^n \in \Omega$, then $g(\freq{x^n}) \geq h$. We can then apply Theorem~\ref{thm:entropyaccumulationext} with the affine tradeoff function $g$ and get the desired result as $\| \nabla g \|_{\infty} \leq \| \nabla f \|_{\infty}$.

The proof for $H^{\varepsilon}_{\max}$ is analogous.
\end{proof}

One natural choice for the event $\Omega$ is that the empirical distribution $\freq{X_1^n}$ takes a particular value~$q$. This yields the following special case of Corollary~\ref{cor:EATconvex}. 

\begin{corollary}\label{cor:EAT}
    Let $\cM_1,\dots, \cM_n$ and  $\rho_{A_1^n B_1^n X_1^n E}$ be such that~\eqref{eq_rhomap} and the Markov conditions~\eqref{eq_Markovgen} hold. Then, for any differentiable and convex min-tradeoff function $f$ for $\cM_1, \ldots, \cM_n$ and for any $q \in \mbP$, we have 
  \begin{align*}
    H_{\min}^\eps(A_1^n | B_1^n E)_{\rho_{|q}} > n f(q) - c \sqrt{n}
  \end{align*} 
  where  $c = 2 \bigl(\log (1+2 d_A) + \left\lceil \| \nabla f(q) \|_\infty \right\rceil  \bigr) \sqrt{1- 2 \log (\varepsilon \rho[q]})$, where $\rho_{|q}$ denotes the state $\rho$ conditioned on the event that $\freq{X_1^n} = q$, and $\rho[q]$ the probability of this event.
\end{corollary}

Note that an analogous statement holds of course for the max-entropy, replacing $f$ by a concave max-tradeoff function and changing the inequality accordingly. 
%

The following corollary specialises the above to the formulation~\eqref{eq_entropyaccumulation}, in which no statistical test is being done, i.e.~the $X_i$ systems are trivial. We provide the statement for the case of the lower boundary.

\begin{corollary} \label{cor_diffi}
    Let $\cM_1,\dots, \cM_n$  and  $\rho_{A_1^n B_1^n E}$ be such that~\eqref{eq_rhomap} and the Markov conditions~\eqref{eq_Markovgen} hold. Then
  \begin{align*}
    H_{\min}^\eps(A_1^n  | B_1^n E)_{\rho} > \sum_ {i=1}^n  \inf_{\omega_{R_{i-1} R}} H(A_i | B_i R)_{(\cM_i \otimes \cI_R)(\omega_{R_{i-1} R})} - c \sqrt{n}
  \end{align*} 
  where  $c = 3 (\log (1+2 d_A) \bigr) \sqrt{1- 2 \log (\varepsilon})$.
\end{corollary}

\begin{proof}
Note that the quantity $H_{\min}^\eps(A_1^n  | B_1^n E)_{\rho}$ only depends on the marginal of the state $\rho$ on $A_1^n B_1^n E$. Thus, we can modify the maps $\cM_i$ in any way that does not affect the reduced state $\rho_{A_1^n B_1^n E}$ before applying Corollary~\ref{cor:EAT}. In particular, we change $\cM_i$ so that the original value of $X_i$ is disregarded and replaced with the constant value $X_i = i$.
The values $X_1, \ldots, X_n$ can then be regarded as random variables with alphabet $\cX = \{1, \ldots, n\}$. We define the real function~$f$ on $\mbP$ as
 \begin{align*}
   f(q) = \sum_{i=1}^n q(i) \inf_{\omega_{R_{i-1} R}} H(A_i | B_i R)_{(\cM_i \otimes \cI_R)(\omega_{R_{i-1} R})} \ .
\end{align*}
Note that for any $i \in \{1,\dots, n\}$ and any $q \in \mbP$, we have either $q(i) \neq 1$ in which case $\Sigma_i(q) = \emptyset$ (we use the notation in~\eqref{eq_def_set_Sigma}) and the min-tradeoff condition is trivial or $q(i) = 1$, in which case $\Sigma_i(q) = \{(\cM_i \otimes \cI_R)(\omega_{R_{i-1} R}) : \omega_{R_{i-1} R} \in \mathrm{D}(R_{i-1} \otimes R) \}$. Thus for any $q \in \mbP$,
\begin{align*}
f(q) \leq \inf_{\omega_{R_{i-1} R}} H(A_i | B_i R)_{(\cM_i \otimes \cI_R)(\omega_{R_{i-1} R})} \ .
\end{align*}
As a result, $f$ is a min-tradeoff function for all $\cM_i$ for $i \in \{1, \dots, n\}$. We now fix $q \in \mbP$ such that $q(1) = \cdots = q(n) = \frac{1}{n}$, in which case the event  $\freq{X_1^n} = q$ occurs with certainty. Because
\begin{align*}
 \| \nabla f(q) \|_\infty \leqslant \log d_A
\end{align*}
which implies that $\left\lceil \| \nabla f(q) \|_\infty \right\rceil \leq \log(1+2d_A)$, the claim follows immediately from Corollary~\ref{cor:EAT}. 
\end{proof}

As indicated in the introduction, in the special case where the individual pairs $(A_i, B_i)$ are independent and identically distributed (IID), the entropy accumulation theorem corresponds to the Quantum Asymptotic Equipartition Property~\cite{TCR09}. We can therefore formulate the latter as a corollary of Theorem~\ref{thm:entropyaccumulationext}.\footnote{In the version of~\cite{TCR09}, the term $1 + 2 d_A$ in the logarithm is replaced by an expression that depends on the entropy of $A$ conditioned on $B$.}
 
\begin{corollary}\label{cor_QAEP}
  For any bipartite state  $\nu_{A B}$, any $n \in \mathbb{N}$, and any $\eps \in (0,1)$, 
  \begin{align*}
  \frac{1}{n} H_{\min}^{\varepsilon}(A_1^n | B_1^n)_{\nu^{\otimes n}}  > H(A|B)_{\nu} - 2 \sqrt{\frac{1-2 \log \eps}{n}} \log(1+2 d_A)  \ .
\end{align*}
\end{corollary}

\begin{proof}
  Let, for any $i=1, \ldots, n$,   $\cM_i$ be the TPCP map from $R$ to $X A B R$ which sets $A B$ to state $\nu_{A B}$ and where $X$ and $R$ are trivial (one-dimensional) systems.  The concatenation of these maps thus generates the state $\rho_{A_1^n B_1^n} = \nu_{A B}^{\otimes n}$. The claim is then obtained from Theorem~\ref{thm:entropyaccumulationext} with the trade-off function $f$ being a constant equal to $h = H(A|B)_{\nu}$ and with $\Omega$ as the certain event.
\end{proof}

%

\section{Applications}\label{sec_applications}

Entropy is a rather general notion and, accordingly, entropy accumulation has applications in various areas of physics, information theory, and computer science. An example from physics is the phenomena of thermalisation. It is known that a system can only thermalise if its smooth min-entropy is sufficiently large~\cite{RelativeThermalisation}.  To illustrate how Theorem~\ref{thm:entropyaccumulationext} could give an estimate of this entropy, consider a system of interest (e.g., a cup of coffee)  which is in contact with a large environment (the air around it). Suppose that, for an appropriately chosen discretisation of the evolution, the system interacts at each time step with a different part of the environment (e.g., with different air molecules bouncing off the coffee cup).\footnote{In statistical mechanics this is formally described by the \emph{Repeated Interaction System (RIS)} model~\cite{AttPau06,BrJoMe14}.} Theorem~\ref{thm:entropyaccumulationext}  then provides a bound on the total entropy that is transferred to the environment in terms of the von Neumann entropy transferred in each time step. Because the joint time evolution of system and environment is unitary, this entropy flow to the environment could be expressed in terms of the entropy change of the system itself. The argument would therefore prove that the total entropy acquired by the system over many time steps is bounded by the sum of the entropies produced in each individual time step.

Another area where the notion of entropy plays a crucial role is quantum cryptography. Many proofs of security of cryptographic protocols involve lower-bounding the uncertainty that a dishonest adversary has about some system of interest. The state-of-the-art is to derive such bounds using a combination of de Finetti-type theorems as well as the Quantum Asymptotic Equipartition Property~\cite{Renner05,Ren07,CKR09,AFR15}. However, the use of de Finetti theorems comes with various disadvantages. Firstly, they are only applicable under certain assumptions on the symmetry of the protocols. Secondly, they introduce additional error terms that can be large in the practically relevant finite-size regime~\cite{ScaRen08}. Finally, it is not known how to apply de Finetti theorems in a device-independent scenario (see~\cite{EkertRenner14} for an overview and references on device-independent cryptography). These problems can all be circumvented by the use of entropy accumulation, as demonstrated in~\cite{AFRV19} for the case of device-independent quantum key distribution and randomness expansion. The resulting security statements  are valid against general attacks and essentially optimal in the finite-size regime. 

In the remainder of this section, we illustrate the use of entropy accumulation with two concrete examples. The first  is a security proof for a basic quantum key distribution protocol.
The second is a novel derivation of an upper bound on the fidelity of fully quantum random access codes.

\subsection{Sample application: Security of quantum key distribution}
A \emph{Quantum Key Distribution (QKD)} protocol enables two parties, Alice and Bob, to establish a common secret key, i.e., a string of random bits unknown to a potential eavesdropper, Eve.  The setting is such that Alice and Bob can communicate over a quantum channel, which may however be fully controlled by Eve. In addition, Alice and Bob have a classical communication link which is assumed to be authenticated, i.e., Eve may read but cannot alter the classical messages exchanged  between Alice and Bob. The protocol is said to be \emph{secure against general attacks} if any attack by Eve is either detected (in which case the protocol aborts) or does not compromise the secrecy of the final key. Here, we will show that our main theorem can be directly applied to show security against general attacks for a fairly standard QKD protocol. As a bonus, our proof still holds even if we do not make any assumptions about Bob's measurement device: the POVM applied by Bob at every step of the protocol can be arbitrary, and may vary from one step to the next (thereby achieving one-sided measurement device independence as in~\cite{TR11}, but without the restriction to memoryless devices; see also~\cite{TFKW13}). In fact, as shown in~\cite{AFRV19}, the entropy accumulation theorem can be used to prove the security of fully device-independent quantum key distribution.

For concreteness, we consider here a variant of the E91 QKD protocol~\cite{Ekert91} (and note that any security proof for this protocol also implies security of the BB84 protocol~\cite{BB84,BBM92}). The protocol consists of a sequence of instructions for Alice and Bob, as described in the box below. These depend on certain parameters, including the number, $n$, of qubits that need to be transmitted over the quantum channel, the maximum tolerated noise level, $e$, of this channel, as well as the key rate, $r$, which is defined as the number of final key bits divided by $n$.  In the first protocol step, Alice and Bob need to measure their qubits at random in one of two mutually unbiased bases, which we term the \emph{computational} and the \emph{diagonal} basis. These are chosen with probability $1-\mu$ and $\mu$, respectively, for some $\mu \in (0,1)$. The protocol also invokes an error correction scheme termed $\mathrm{EC}$, which allows Bob to infer the measurement outcomes obtained by Alice for the set of indices $S$ where the basis choices of Alice and Bob were the same. Note that if the protocol was implemented without any noise, then Bob's outcomes would match exactly with Alice's outcomes on the indices $S$ and no error correction would be required. However, in the presence of noise, such an error correction step is needed. For this, Alice needs to send classical error correcting information to Bob, whose maximum relative length is characterised by another parameter, $\vartheta_{\mathrm{EC}}$. We assume that $\mathrm{EC}$ is \emph{reliable}. This means that, except with negligible probability, Bob either obtains a correct copy of Alice's string or he is notified that the string cannot be inferred.\footnote{Any error correction scheme can be turned into a reliable one by appending a test where Alice and Bob compare a hash value computed from their (corrected) strings.}  

\begin{table}[h]
\begin{center}
\fbox{\begin{minipage}{0.96\linewidth}
\smallskip
{\bf The E91 Quantum Key Distribution Protocol}

\smallskip

\underline{Protocol parameters}  \vspace{-0.8ex}
\vspace{-0.8ex}
\begin{center}
\begin{tabular}{r c l}
  $n \in \mathbb{N}$ & : & number of uses of qubit channel \\
   $\mu \in (0,1)$ & : & probability for measurements in diagonal basis \\
  $e \in (0,\frac{1}{2})$ & : & maximum tolerated phase error ratio  \\
  $\vartheta_{\mathrm{EC}} \in [0,1]$ & : & relative communication cost of error correction scheme $\mathrm{EC}$ \\
  $r \in [0,1]$ & : & key rate  
\end{tabular}
\end{center}

\vspace{0.1ex}

{\underline{Protocol steps}}
\vspace{-1ex}
\begin{enumerate}
    \setlength{\itemsep}{0.4ex}
    \setlength{\parskip}{0.4ex}
    \setlength{\parsep}{0.4ex} 
  \item \emph{Distribution:} For $i \in \{1, \ldots, n\}$,  Alice prepares a pair $(Q_i, \bar{Q}_i)$ of entangled qubits and sends $\bar{Q}_i$ to Bob. Alice generates a random bit $B_i$ such that $P_{B_i}(1) = \mu$ and, depending on whether $B_i=0$ or $B_i = 1$, measures $Q_i$ in either the computational or  the diagonal basis, storing the outcome as $A_i$. In the same way, Bob measures $\bar{Q}_i$ in a basis determined by a random bit $\bar{B}_i$, storing the outcome as $\bar{A}_i$. 
  
  \item \emph{Sifting and information reconciliation:} Alice and Bob announce $B_i$ and $\bar{B}_i$ and determine the set $S$ of indices $i \in \{1, \ldots, n\}$ such that $B_i = \bar{B}_i$.  They invoke the error correction scheme $\mathrm{EC}$, allowing Bob to compute a guess $\hat{A}_S$ for Alice's string $A_{S} = (A_i)_{i \in S}$.  If $\mathrm{EC}$ does not output a guess then the protocol is aborted. 
  
  \item \emph{Parameter estimation:}  Bob counts the number of indices $i \in S$  for which $\bar{B}_i = 1$ and $\bar{A}_i \neq \hat{A}_i$. If this number is larger than $e \mu^2 n$ then the protocol is aborted.

  \item \emph{Privacy amplification:} Alice chooses a function $F$ at random from a two-universal set of hash functions~\cite{WegmanCarter} from $|S|$ bits to $\lfloor r n \rfloor$ bits and announces $F$ to Bob. Both Alice and Bob  compute the final key as $F(A_S)$ and $F(\hat{A}_S)$, respectively. 
  \end{enumerate} 
\vspace{-1ex}
\end{minipage} \hspace{1em}}
\end{center} 
 \end{table}
 
The security of QKD against general attacks has been established in a sequence of works~\cite{lo&chau:security,Mayers01,SP00,Biham00,Renner05}. Specifically, for the E91 protocol, the following result has been shown.

\begin{theorem} \label{thm_QKD}
  The E91 protocol is secure for any choice of protocol parameters satisfying\footnote{$H_{\mathrm{Sh}}(e) = -e \log e - (1-e) \log (1-e)$ is the binary Shannon entropy.} 
\begin{align} \label{eq_secrecycondition}
  r <  1 - H_{\mathrm{Sh}}(e)  - \vartheta_{\mathrm{EC}} - 2 \mu \ ,
\end{align}
provided that $n$ is sufficiently large. 
\end{theorem}

Note that, because $\mu > 0$ can be chosen arbitrarily small, the theorem implies that the E91 protocol can generate secret keys at an asymptotic rate of $1-H_{\mathrm{Sh}}(e) - \vartheta_{\mathrm{EC}}$.  We now show how this result can be obtained using the notion of entropy accumulation.

\begin{proof}
According to a standard result on two-universal hashing (see, for instance, Corollary~5.6.1 of~\cite{Renner05}), the key $F(A_S)$ computed in the privacy amplification step is secret to an adversary holding information $E'$ if the smooth min-entropy of $A_S$ conditioned on $E'$ is sufficiently larger than the output size of the hash function $F$. Since, in our case, this size is $\lfloor r n \rfloor$,  the condition reads  
\begin{align} \label{eq_PACriterion}
n r \leqslant H_{\min}^{\eps}(A_S | E')_{\rho_{|\Omega}} - O(1) \ ,
\end{align}
where the entropy is evaluated for the joint state $\rho_{|\Omega}$ of $A_S$ and $E'$ conditioned on the event $\Omega$ that the protocol is not aborted and that Bob's guess $\hat{A}_S$ of $A_S$ is correct. The smoothing parameter $\eps \in (0,1)$ specifies the desired level of secrecy,\footnote{Roughly, $\eps$ corresponds to the maximum probability by which one could encounter a deviation from perfect secrecy~\cite{PortmannRenner}.} and we assume here that it is constant (independent of $n$).  Because conditioning the smooth min-entropy of a classical variable on an additional bit cannot decrease its value by more than~$1$ (see, e.g., Proposition~5.10 of~\cite{Tom12}), we may bound the smooth min-entropy in~\eqref{eq_PACriterion} by 
\begin{align} \label{eq_ECbound}
  H_{\min}^\eps(A_S | E')_{\rho_{|\Omega}} \geqslant H_{\min}^\eps(A_S | B_1^n \bar{B}_1^n E)_{\rho_{|\Omega}} - |S| \vartheta_{\mathrm{EC}}
   \geqslant H_{\min}^\eps(A_S | B_1^n \bar{B}_1^n E)_{\rho_{|\Omega}} - n \vartheta_{\mathrm{EC}} \ ,
\end{align}
where $E$ denotes all information held by Eve after the distribution step, and where $|S| \vartheta_{\mathrm{EC}}$ is the maximum number of bits exchanged for error correction.  Note that we also included the basis information $B_1^n$ and $\bar{B}_1^n$ in the conditioning part because Eve may obtain this information during the sifting and information reconciliation step.  We are thus left with the task of lower bounding $H_{\min}^\eps(A_S|B_1^n \bar{B}_1^n E)_{|\rho_{|\Omega}}$, which is usually the central part of any security proof. Since it is also the part where entropy accumulation is used, we formulate it separately as Claim~\ref{claim_QKD} below. Inserting this claim into~\eqref{eq_ECbound}, we conclude that the secrecy condition~\eqref{eq_PACriterion} is fulfilled whenever
\begin{align*}
 n r  \leqslant  n \bigl(1-H_{\mathrm{Sh}}(e) - \vartheta_{\mathrm{EC}} - 2\mu \bigr) - o(n) 
\end{align*}
holds. But this is clearly the case for any choice of parameters satisfying~\eqref{eq_secrecycondition}, provided that~$n$ is sufficiently large. 
\end{proof}

It remains to show the separate claim, which we do using entropy accumulation. 

\begin{claim} \label{claim_QKD}
Let $A_1^n$, $B_1^n$, $\bar{B}_1^n$, and $S$ be the information held by Alice and Bob as defined by the protocol, let $E$ be the information gathered by Eve during the distribution step, and let $\Omega$ be the event that the protocol is not aborted and that Bob's guess $\hat{A}_S$ of $A_S$ is correct. Then, provided that $\Omega$ has a non-negligible probability (i.e., it does not decrease exponentially fast in $n$), 
\begin{align} \label{eq_QKDtoprove}
  H_{\min}^\eps(A_S | B_1^n \bar{B}_1^n E)_{\rho_{|\Omega}} > n \bigl(1 - 2\mu - H_{\mathrm{Sh}}(e)\bigr) - o(n) \ .
\end{align}
\end{claim}

\begin{proof}
Let $\rho^0_{Q_1^n \bar{Q}_1^n E}$ be the joint state of Alice and Bob's qubit pairs before measurement, together with the information $E$ gathered by Eve during the distribution step, and let
\begin{align*}
  \rho_{A_1^n \bar{A}_1^n B_1^n \bar{B}_1^n X_1^n E} = (\cM_n \circ \cdots \circ \cM_1 \otimes \cI_E)(\rho^0_{Q_1^n \bar{Q}_1^n E}) \ ,
\end{align*}
where $\cM_i$, for any $i \in \{1, \ldots, n\}$, is the TPCP map from $Q_i^n \bar{Q}_i^n$ to $Q_{i+1}^n \bar{Q}_{i+1}^n A_i \bar{A}_i B_i \bar{B}_i X_i$ defined as follows: 
\begin{enumerate}[(i)]
  \item $B_i, \bar{B}_i$: random bits chosen independently according to the distribution $(1-\mu, \mu)$
  \item $A_i = \begin{cases} \text{if } B_i = \bar{B}_i = 0: & \text{outcome of measurement of $Q_i$ in computational basis} \\ \text{if } B_i = \bar{B}_i = 1: & \text{outcome of measurement of $Q_i$ in diagonal basis} \\  \text{if } B_i \neq \bar{B}_i : &  \perp \end{cases}$
    \item $\bar{A}_i = \begin{cases} \text{if } B_i = \bar{B}_i = 1 : & \text{outcome of measurement of $\bar{Q}_i$ in diagonal basis} \\ \text{otherwise} : &  \perp \end{cases}$
  \item $X_i = \begin{cases} \text{if } B_i = \bar{B}_i = 1 : & A_i \oplus \bar{A}_i \\ \text{otherwise} : & \perp \end{cases}$
  \item $Q_{i+1}^n$ and $\bar{Q}_{i+1}^n$ are left untouched. 
\end{enumerate}
Note that the values $B_1^n$ and $\bar{B}_1^n$ correspond to the ones generated during the distribution step of the protocol. The same is true for $A_1^n$, with the modification that $A_i$ holds the measurement outcome only if $B_i = \bar{B}_i $. That is, $A_i \neq \perp$ if and only if $i \in S$, where $S$ is the set determined in the sifting step. We can therefore rewrite~\eqref{eq_QKDtoprove} as
\begin{align} \label{eq_QKDtoproveprime}
  H_{\min}^\eps(A_1^n | B_1^n \bar{B}_1^n E)_{\rho_{|\Omega}} > n \bigl(1 - 2\mu - H_{\mathrm{Sh}}(e)\bigr) - o(n) \ .
\end{align}

To prove this inequality, we use Theorem~\ref{thm:entropyaccumulationext} with the replacements $A_i \rightarrow A_i \bar{A}_i$, $B_i \rightarrow B_i \bar{B}_i$, $X_i \rightarrow X_i$, and $R_i \rightarrow Q_{i+1}^n \bar{Q}_{i+1}^n$. We note that $X_i$ is a deterministic function of the classical registers $A_i \bar{A}_i$ and $B_i \bar{B}_i$. To obtain the bound in~\eqref{eq_QKDtoproveprime}, we need to define a min-tradeoff function. Let $i \in \{1, \ldots, n\}$ and consider the state
\begin{align*}
  \nu_{X_i A_i \bar{A}_i B_i \bar{B}_i R} = \tr_{Q_{i+1}^n \bar{Q}_{i+1}^n}(\cM_i \otimes \cI_R)(\omega_{Q_{i}^n \bar{Q}_{i}^n R}) \ , 
\end{align*}
where  $\omega_{Q_{i}^n \bar{Q}_{i}^n R}$ is an arbitrary state. Let furthermore $\nu_{|b} = \nu_{X_i A_i \bar{A}_i R | b}$ be the corresponding state obtained by conditioning on the event that $B_i = \bar{B}_i = b$, for $b \in \{0, 1\}$. We may now bound the entropy of $A_i$ using the entropic uncertainty relation proved in~\cite{BCCRR10}, which asserts that
\begin{align*}
  H(A_i | R)_{\nu_{|0}} \geqslant 1 - H(A_i | \bar{A}_i)_{\nu_{|1}} \ .
\end{align*}
By the definition of $X_i$, we also have
\begin{align*}
  H(A_i | \bar{A}_i)_{\nu_{|1}} = H(X_i)_{\nu_{|1}} = H_{\mathrm{Sh}}\left({\textstyle \frac{\nu_{X_i}(1)}{\nu_{X_i}(0) + \nu_{X_i}(1)}}\right)  =  H_{\mathrm{Sh}}\left({\textstyle \frac{\nu_{X_i}(1)}{\mu^2}}\right) \ ,
\end{align*}
where we wrote $\nu_{X_i}$ to denote the probability distribution on $\{0, 1, \bot\}$ defined by the state~$\nu$, and where we have used that $\nu_{X_i}(0) + \nu_{X_i}(1) = \mu^2$. Furthermore, because $A_i$ is classical, its von Neumann entropy cannot be negative, which implies that
\begin{align*}
  H(A_i | B_i \bar{B}_i R)_{\nu} \geqslant \nu_{B_i \bar{B}_i}(0,0) H(A_i | R)_{\nu_{|0}} = (1-\mu)^2 H(A_i | R)_{\nu_{|0}} \geqslant  H(A_i | R)_{\nu_{|0}}  - 2 \mu + \mu^2 \ .
\end{align*}
Combining this with the above, we find that
\begin{align*}
  H(A_i \bar{A}_i | B_i \bar{B}_i R)_{\nu} \geqslant H(A_i | B_i \bar{B}_i R)_{\nu} \geqslant  \tilde{f}(\nu_{X_i}) 
\end{align*}
holds for
\begin{align*}
  \tilde{f}(q) = \begin{cases} 1 - 2 \mu + \mu^2 - H_{\mathrm{Sh}}\left({\textstyle \frac{q(1)}{\mu^2}} \right)  & \text{if $q(0) + q(1) = \mu^2$} \\ 1 & \text{otherwise.} \end{cases}
\end{align*}
In other words, $\tilde{f}$ is a min-tradeoff function for $\cM_i$. Furthermore, because the binary Shannon entropy $H_{\mathrm{Sh}}$ is concave, $\tilde{f}$ is convex. We may thus define a linearised min-tradeoff function $f$ as a tangent hyperplane to $\tilde{f}$ at the point $q_0$ given by $q_0(0) = (1-e) \mu^2$, $q_0(1) = e \mu^2$, and $q_0(\bot) = 1-\mu^2$.  Furthermore, we define 
\begin{align*}
  h = f(q_0) = \tilde{f}(q_0) = 1 - 2 \mu + \mu^2  - H_{\mathrm{Sh}}(e) \ .
\end{align*}
Finally, note that the event $\Omega$ that Bob's guess of $A_S$ is correct and that the protocol is not aborted implies that $q = \freq{X_1^n}$ is such that $\frac{q(1)}{\mu^2} \leqslant e$ and, hence, $f(\freq{X_1^n}) \geqslant h$. Since we assumed that $\Omega$ has non-negligible probability, Theorem~\ref{thm:entropyaccumulationext} implies that
\begin{align*}
  H_{\min}^{\eps/4}(A_1^n \bar{A}_1^n | B_1^n \bar{B}_1^n E)_{\rho_{|\Omega}} > n h - o(n) = n \bigl(1 - 2\mu + \mu^2 - H_{\mathrm{Sh}}(e)\bigr) - o(n) \ .
\end{align*}
(Note that the Markov chain conditions are satisfied because $B_i$ and $\bar{B}_i$ are chosen at random independently of any other information.) Furthermore, because $\bar{A}_i$ equals $\bot$ unless $B_i = \bar{B}_i = 1$, which occurs with probability~$\mu^2$, we have 
\begin{align*}
  H_{\max}^{\frac{\eps}{4}}(\bar{A}_1^n | A_1^n B_1^n \bar{B}_1^n E)_{\rho_{|\Omega}} 
  \leqslant H_{\max}^{\frac{\eps}{4}}(\bar{A}_1^n | B_1^n \bar{B}_1^n)_{\rho_{|\Omega}} 
\leqslant \mu^2 n + o(n) \ .
\end{align*}
 Combining these inequalities with the chain rule for smooth entropies (see Theorem~15 of~\cite{VDTR13}),
\begin{align*}
  H_{\min}^\eps(A_1^n | B_1^n \bar{B}_1^n E)_{\rho_{|\Omega}}
  \geqslant H_{\min}^{\eps/4}(A_1^n \bar{A}_1^n| B_1^n \bar{B}_1^n E)_{\rho_{|\Omega}} - H_{\max}^{\eps/4}(\bar{A}_1^n | A_1^n B_1^n \bar{B}_1^n E)_{\rho_{|\Omega}} - O(1) \ ,
\end{align*}
proves~\eqref{eq_QKDtoproveprime} and, hence, Claim~\ref{claim_QKD}.
\end{proof}

\subsection{Sample application: Fully quantum random access codes}
\label{sec:fqrac}

One relatively simple application of our main result is to give upper bounds on the fidelity achieved by so-called \emph{Fully Quantum Random Access Codes (FQRAC)}. An FQRAC is a method for encoding $m$ message qubits into $n < m$ code qubits, such that any subset of $k$ message qubits can be retrieved with high fidelity. Limits on the performance of random access codes with classical messages are rather well understood: the case $k=1$ was studied in \cite{n99,antv99,antv02}, and upper bounds on the success probability that decay exponentially in $k$ were derived in \cite{bARdW08,w10,DFW13}. In the fully quantum case, \cite{DFW13} gives similar upper bounds on the fidelity that decay exponentially in $k$. Here, we show that such exponential bounds for the fully quantum case can be obtained in a relatively elementary fashion via the concept of entropy accumulation.  The example also highlights that entropy accumulation is already useful in its basic form~\eqref{eq_entropyaccumulation}, which does not involve statistics information $X_i$. Indeed, here the bound on the entropy produced at every step comes from the bound on the number of code qubits.

\begin{definition} A \emph{$(\varepsilon,m,n,k)$-Fully Quantum Random Access Code (FQRAC)} consists of an encoder $\mathcal{E}_{{M'}_1^m \rightarrow C_1^n}$ and a decoder $\mathcal{D}_{C_1^n S \rightarrow \bar{M}_S S}$, where ${M'}_1^m$ represents the $m$ message qubits, $C_1^n$ represents the $n$ code qubits, $S$ represents a classical description of a subset of $\{1,\dots,m\}$ of size $k$, and $\bar{M}_S$ represents the output of the decoder, corresponding to the $k$ positions of ${M'}_1^m$ listed in $S$. Such a code must satisfy the following: for any state $\rho_{R {M'}_1^m S}$ which is classical on $S$, we must have that
    \[ F\big(\mathcal{S}(\rho_{R {M'}_1^m S}), (\mathcal{D} \circ \mathcal{E})(\rho_{R {M'}_1^m S}) \big)^2 \geqslant 1 - \varepsilon, \]
    where $R$ is a reference system of arbitrary dimension, and where $\mathcal{S}_{ {M'}_1^m S \rightarrow \bar{M}_S S}$ is a TPCP map that selects the $k$ positions of ${M'}_1^m$ corresponding to those in $S$ and outputs them into~$\bar{M}_S$. Moreover, $F(\rho,\sigma) := \| \sqrt{\rho} \sqrt{\sigma}\|_1$ refers to the fidelity between two states $\rho$ and $\sigma$.
\end{definition}

Entropy accumulation gives the following constraint on FQRACs:

\begin{theorem}
    A $(\varepsilon,m,n,k)$-FQRAC satisfies
    \begin{align} \label{eq_QRACbound}
       1-\eps = f^2 < 2^{-k \left(\frac{m-n-k+1}{ 5 m} \right)^2 + 3} \ .
    \end{align}
 \end{theorem}
 
Compared to the previously derived bound (Theorem~9 of~\cite{DFW13}), the one obtained here is tighter for small~$k$,\footnote{The bound of Theorem~9 of~\cite{DFW13} has a pre-factor of the order of~$m$ and is therefore only non-trivial if $k > \log m$.} whereas it is weaker for large~$k$. 

\begin{proof}
    Since the fidelity bound must be true for any state $\rho$, it must in particular be true for the state consisting of $m$ maximally entangled pairs and a uniform distribution over subsets~$S$. For every $i \in \{1,\dots,k\}$, define
   \begin{align*}
      \cM_i : M_1^{m-i+1} \rightarrow M_1^{m-i} \bar{J}_i \hat{M}_i 
   \end{align*}
    as a TPCP map that does the following:
    \begin{enumerate}
        \item Generate an index $\bar{J}_i$ at random from $\{ 1,\dots,m - i +1 \}$.
        \item Move the contents of $M_{\bar{J}_i}$ into $\hat{M}_i$, and set $M_1^{m-i}$ to the contents of $M_1^{m-i+1}$ with the $\bar{J}_i$th position removed.
    \end{enumerate}
    Finally, define the state
    \[ \rho^i_{M_1^{m-i} \hat{M}_1^i C_1^n \bar{J}_1^i} := (\cM_i \circ \dots \circ \cM_1 \circ \mathcal{E})(\proj{\Phi^{\otimes m}}_{M_1^m {M'}_1^m}), \]
    where $\ket{\Phi}_{M_i M_i'} := \frac{1}{\sqrt{2}}\left( \ket{00} + \ket{11} \right)$. The next step is to use Theorem~\ref{thm:entropyaccumulationext} on the state $\rho^k$ with the identifications 
  \begin{align*}  
    A_i \rightarrow \hat{M}_i \qquad  B_i \rightarrow \bar{J}_i \qquad E \rightarrow C_1^n 
 \end{align*}
   and the tradeoff function $f$ being the constant function equal to 
    \begin{align*}
    \inf_{i, \nu^i} H(\hat{M}_i | \hat{M}_1^{i-1} \bar{J}_1^i C_1^n)_{\nu^i} \ ,
        \end{align*} where the infimum is taken over states $\nu^i$ of the form
        \begin{align*}
      \nu^i_{\hat{M}_1^{i} \bar{J}_1^i C_1^n} = \cM_i\left(\omega^i_{\hat{M}_1^{i-1} \bar{J}_1^{i-1} M_1^{m-i+1} C_1^n} \right) \ ,
    \end{align*}
    for some state $\omega^i$. Here we also used Remark~\ref{rem_Rdimensiona}, which asserts that the system $R$ that is used when defining the min-tradeoff function can be chosen isomorphic to $A_1^{i-1} B_1^{i-1} E$.    Note that the Markov chain condition is immediate from the fact that $\bar{J}_i$ is chosen at random. As the systems $X_i$ are trivial, we naturally take $\Omega$ to be the certain event.
     We find that
    \begin{align*}
        H_{\min}^{f/2}(\hat{M}_1^k | \bar{J}_1^k C_1^n)_{\rho^k} &\geqslant k \inf_{i, \nu^i} H(\hat{M}_i|\hat{M}_1^{i-1} \bar{J}_1^i C_1^n)_{\nu^i} - \sqrt{4 k \log \frac{8}{f^2}}  \log 5 \ .
    \end{align*}
     Furthermore, again by Remark~\ref{rem_Rdimensiona}, if part of $B$ is classical in $\rho$, then it remains classical in~$\nu$.  As a result, we can assume in the following that $\bar{J}_{1}^{i-1}$ is a classical system in $\nu^i$. 

    We continue by computing the expectation over the choice of $\bar{J}_i$:
    \begin{align}
        H(\hat{M}_i|\hat{M}_1^{i-1} \bar{J}_1^i C_1^n)_{\nu^i} &= \frac{1}{m-i+1} \sum_{j_i = 1}^{m-i+1} H(M_{j_i} | \hat{M}_1^{i-1} C_1^n \bar{J}_1^{i-1})_{\omega^i}\\
        &\geqslant \frac{1}{m} \sum_{j_i = 1}^{m-i+1} H(M_{j_i} | M_1^{j_i-1} \hat{M}_1^{i-1} C_1^n \bar{J}_1^{i-1})_{\omega^i}\\
        &= \frac{1}{m} H(M_{1}^{m-i+1} | \hat{M}_1^{i-1} C_1^n \bar{J}_1^{i-1})_{\omega^i}\\
        & = \frac{1}{m}  \left( H(M_{1}^{m-i+1} \hat{M}_1^{i-1} C_1^n | \bar{J}_1^{i-1})_{\omega^i} - H(\hat{M}_1^{i-1} C_1^n | \bar{J}_1^{i-1})_{\omega^i}\right) \\
        &\geqslant \frac{-n-k+1}{m} \ ,
    \end{align}
    where the last inequality holds because $\bar{J}_i$ is classical, which implies that the first entropy in the bracket of the penultimate expression is non-negative, and because the second entropy in the bracket is upper bounded by $n+k-1$. 
    
We now use Proposition~5.5 and Remark~5.6 of~\cite{Tom12}, which imply that\footnote{Because $\arcsin(f/2) + \arcsin(\sqrt{1-f^2}) < \arcsin(f) + \arcsin(\sqrt{1-f^2}) = \pi/2$, the condition of Remark~5.6 of~\cite{Tom12} is satisfied.}
\begin{multline*}
  H_{\max}^{\sqrt{1-f^2}}(\hat{M}_1^k | \bar{J}_1^k C_1^n)_{\rho^k}
  \geqslant H_{\min}^{f/2}(\hat{M}_1^k | \bar{J}_1^k C_1^n)_{\rho^k} - \log \frac{1}{1-\bigl(f^2/2  + \sqrt{1-f^2} \sqrt{1-f^2/4}\bigr)^2} \\
    \geqslant H_{\min}^{f/2}(\hat{M}_1^k | \bar{J}_1^k C_1^n)_{\rho^k} - \log \frac{3}{f^3} \ ,
\end{multline*}
where the second inequality holds because the denominator in the logarithm is lower bounded by $f^3/3$, as can be readily verified. Combining this with the above gives
\begin{align*}
  H_{\max}^{\sqrt{1-f^2}}(\hat{M}_1^k | \bar{J}_1^k C_1^n)_{\rho^k}
  \geqslant - k \left( \frac{n+k-1}{m} \right)  - \sqrt{4 k \log \frac{8}{f^2}}  \log 5  - \log \frac{3}{f^3}  \ .
\end{align*}
Conversely, note that, by assumption, the purified distance between $\rho^k$ and the state consisting of $k$ maximally entangled qubit pairs is upper bounded by $\sqrt{1-(1-\eps)} = \sqrt{1-f^2}$. Since the max-entropy of $k$ maximally entangled qubit pairs equals $-k$, we have 
\begin{align*}
  H_{\max}^{\sqrt{1-f^2}}(\hat{M}_1^k | \bar{J}_1^k C_1^n)_{\rho^k} \leqslant -k \ .
  \end{align*}
  We have thus derived the condition
\begin{align}
\label{eq_condition_fqrac}
  \sqrt{4 k \log \frac{8}{f^2}}  \log 5
  \geqslant k \left( \frac{m-n-k+1}{m} \right) -  \log \frac{3}{f^3}   \ .
\end{align}
It is easy to verify that this condition is violated whenever 
\begin{align}
\label{eq_necessary_cond_fqrac}
  \log \frac{8}{f^2} > k \left(\frac{m-n-k+1}{ 5 m} \right)^2 
\end{align}
is violated. In fact, if $\log \frac{8}{f^2} \leqslant k \left(\frac{m-n-k+1}{ 5 m} \right)^2$, then we have
\begin{align*}
4k \log \frac{8}{f^2} \log^2 5 &\leqslant \frac{4 \log^2 5}{25} k^2 \left(\frac{m-n-k+1}{ m} \right)^2 \ , \text{ and } \\
\log \frac{3}{f^3} &\leqslant \frac{3}{2} \log \frac{8}{f^2} \leqslant \frac{3}{50} k \left(\frac{m-n-k+1}{ m} \right)^2 \leqslant \frac{3}{50} k \left(\frac{m-n-k+1}{ m} \right) \ .
\end{align*}
Adding the square root of the first inequality and the second one, we get that inequality~\eqref{eq_condition_fqrac} is violated. Thus, the condition~\eqref{eq_necessary_cond_fqrac} must hold, and therefore also~\eqref{eq_QRACbound}.
\end{proof}

\section{Conclusions} \label{sec_conclusions}

Informally speaking, entropy accumulation is the claim that the operationally relevant entropy (the smooth min- or max-entropy) of a multipartite system is well approximated by the sum of the von Neumann entropies of its individual parts.  This has ramifications in various areas of science, ranging from quantum cryptography to thermodynamics. 

As described in Section~\ref{sec_applications}, current cryptographic security proofs have various fundamental and practical limitations~\cite{QKDSecurityReview}. That these can be circumvented using entropy accumulation has already been demonstrated in~\cite{AFRV19} for the case of device-independent cryptography. We anticipate that the approach can be applied similarly to other cryptographic protocols.  Examples include quantum key distribution protocols  such as DPS and COW~\cite{InoHon05,COWProtocol2005}, for which full security has  not yet been established.\footnote{These protocols do not have the required symmetries to employ standard techniques such as de Finetti-type theorems~\cite{Ren07}.} One may also expect to obtain significantly improved security bounds for protocols that involve high-dimensional information carriers and, in particular,  continuous-variable protocols~\cite{GroGra02,Weedbrooketal04}.\footnote{The security of continuous-variable protocols against general attacks has been proved~\cite{CirRen09}, but the bounds have an unfavourable scaling in the finite-size regime.}  A strengthening of current security claims may as well be obtained for other cryptographic constructions, such as bit commitment and oblivious transfer protocols (see, for example, \cite{serge:bounded,KWW09,DFW13}). 

Entropy accumulation can also be used in statistical mechanics, e.g., to characterise thermalisation processes.  At the beginning of Section~\ref{sec_applications} we outlined an argument that could confirm | and make precise | the intuition that entropy production (in terms of von Neumann entropy) is relevant for thermalisation. However, to base such arguments on physically realistic assumptions, it may be necessary to generalise Theorem~\ref{thm:entropyaccumulationext} to the case where the Markov conditions~\eqref{eq_Markovgen} do not hold exactly. One possibility, motivated by the main result of~\cite{FR2015}, could be to replace them by the less stringent conditions 
  \begin{align} \label{eq_MarkovRelaxed}
    H(B_i | B_1^{i-1} E) \approx  H(B_i | A_1^{i-1} B_1^{i-1} E)  \ .
  \end{align}
  Another promising direction would be to apply entropy accumulation to estimate the entropy of  low-energy states of many-body systems. One may expect that, under appropriate physical assumptions, these states possess a structure that permits a decomposition of the form described by Figure~\ref{fig:rho-structure} such that the Markov conditions required for Theorem~\ref{thm:entropyaccumulationext}, or at least some relaxations of them such as~\eqref{eq_MarkovRelaxed}, hold. This may for example be the case for systems whose states are well approximated by matrix products states (see, e.g., \cite{VerCir06}).  We leave the investigation of such applications, as well as the development of corresponding extensions of the entropy accumulation theorem, for future work. 

\begin{appendices}

\section{The function $\| \cdot \|_\alpha$} \label{app_alpha}

We use an extension of the Schatten $\alpha$-norm to the regime where $\alpha > 0$, which is defined for any operator $X = X_{B \leftarrow A}$ from a space $A$ to a space $B$ by
\begin{align*}
  \left\| X \right\|_\alpha = \tr\left( (X^{\dagger} X)^{\frac{\alpha}{2}} \right)^{\frac{1}{\alpha}} \ .
\end{align*}
It follows from the Singular Value Theorem that $\| X \|_\alpha = \| X^{\dagger} \|_\alpha = \| X^\intercal \|_{\alpha} =  \| \overline{X} \|_{\alpha}$  (see, e.g,  Section~2 of~\cite{WatrousQILectures}), from which it also follows that 
\begin{align*}
    \left\| X \right\|_\alpha = \tr\left( (X X^{\dagger})^{\frac{\alpha}{2}} \right)^{\frac{1}{\alpha}} \ .
\end{align*}
Note also that
\begin{align} \label{eq_alphatwo}
      \left\| X \right\|_{2 \alpha}^2 =  \left\| X^{\dagger} X \right\|_{\alpha} =  \left\| X X^{\dagger} \right\|_{\alpha} \ .
\end{align}

The following is Lemma~12 from~\cite{MDSFT13}.

\begin{lemma} \label{lem_normdual}
  For any non-negative operator $X$ and for any $\alpha \in \mathbb{R}^+$
  \begin{align} \label{eq_normdual}
    \left\| X \right\|_\alpha = \begin{cases} \sup_Z \tr(X Z^{\alpha'})  & \text{if $\alpha' \geqslant 0$} \\  \inf_Z  \tr(X Z^{\alpha'}) & \text{if $\alpha' \leqslant 0$} \end{cases} 
  \end{align}
  where  the supremum and infimum range over density operators $Z$. 
\end{lemma}

\section{Properties of the sandwiched R\'enyi entropies} \label{app_sandwich}

The sandwiched R\'enyi entropy from Definition~\ref{def_sandwichedentropy} is a special case of the sandwiched R\'enyi relative entropy, which is defined as follows.

\begin{definition} \label{def_sandwichedrelativeentropy}
  For two density operators $\rho$ and $\sigma$ on the same Hilbert space and for $\alpha \in (0, 1) \cup (1, \infty)$ the \emph{sandwiched relative R\'enyi entropy of order $\alpha$} is defined as
  \begin{align*}
    D_{\alpha}(\rho \| \sigma) = \frac{1}{\alpha'} \log \left\| \rho^{\frac{1}{2}}  \sigma^{\frac{-\alpha'}{2}} \right\|_{2 \alpha}^2  \ ,
   \end{align*}
   where $\alpha' = \frac{\alpha - 1}{\alpha}$.
\end{definition}

In particular, for a bipartite density operator $\rho_{A B}$, the sandwiched $\alpha$-R\'enyi entropy of $A$ conditioned on $B$ is related to this relative entropy by
\begin{align} \label{eq_condrelentropy}
  H_{\alpha}(A|B)_{\rho} = - D_{\alpha}(\rho_{A B} \| \id_A \otimes \rho_B) \ .
\end{align}
It turns out that this is not the only way to define a conditional entropy based on a relative entropy. One popular alternative is to replace the marginal $\rho_B$ by a maximisation over arbitrary density operators on $B$:
\begin{align} \label{eq_condrelentropy_uparrow}
    H^{\uparrow}_{\alpha}(A|B)_{\rho} = - \inf_{\sigma_B} D_{\alpha}(\rho_{A B} \| \id_A \otimes \sigma_B) \ .
\end{align}
We refer to~\cite{TomamichelBertaHayashi} for a comparison of the different notions.

The following Lemma corresponds to Eq.~19 of~\cite{MDSFT13}. For its proof, it is convenient to represent vectors of product systems as matrices. Let $\{\ket{i}_A\}$ and $\{\ket{j}_B\}$ be fixed orthonormal bases of $A$ and $B$, respectively. For any vector
  \begin{align*}
    \ket{\psi }_{A B} = \sum_{i, j} \gamma_{i,j} \ket{i}_A \otimes \ket{j}_B \in A \otimes B 
  \end{align*}
  we define the linear operator 
  \begin{align*}
    \mathrm{op}_{B \leftarrow A}(\ket{\psi}) =  \sum_{i,j} \gamma_{i,j} \ket{j}_B \bra{i}_A  \ .
  \end{align*}
  We emphasise that this definition is basis-dependent. Therefore, in expressions that involve this operator as well as the transpose operation $Z \mapsto Z^{\intercal}$, it is understood that both are taken with respect to the same basis. It is straightforward to prove the following properties (see, e.g., Section~2.4 of~\cite{WatrousQILectures}). For any operators $X_{A' \leftarrow A}$ and $Y_{B' \leftarrow B}$,
   \begin{align} \label{eq_opinc}
      Y_{B' \leftarrow B} \mathrm{op}_{B \leftarrow A}(\ket{\psi})  X_{A' \leftarrow A}^\intercal
     =  \mathrm{op}_{B \leftarrow A}\bigl((X_{A' \leftarrow A} \otimes Y_{B' \leftarrow B}) \ket{\psi} \bigr) \ .
   \end{align}
   Furthermore, 
      \begin{align} \label{eq_oppart}
            \mathrm{op}_{B \leftarrow A}(\ket{\psi}) \mathrm{op}_{B \leftarrow A} (\ket{\psi})^{\dagger} = \tr_A(\proj{\psi}) 
          \quad \text{and} \quad
             \mathrm{op}_{B \leftarrow A}(\ket{\psi})^{\dagger} \mathrm{op}_{B \leftarrow A}(\ket{\psi}) =    \tr_B(\proj{\psi})^\intercal
   \end{align}
   and, hence, 
   \begin{align} \label{eq_opnorm}
     \left\| \mathrm{op}_{B \leftarrow A}(\ket{\psi}) \right\|_2 = \| \ket{\psi} \| \ .
   \end{align}

\begin{lemma} \label{lem_Halphavector}
  For any density operators $\rho$ and $\sigma$ on the same Hilbert space and for $\alpha \in (0, 1) \cup (1, \infty)$ we have
  \begin{align*}
    D_{\alpha}(\rho \| \sigma) = \sup_{\tau} \frac{1}{\alpha'} \log  \left\| \bigl( \sigma^{\frac{-\alpha'}{2}} \otimes \tau^{\frac{\alpha'}{2}} \bigr) \ket{\psi} \right\|^2
  \end{align*}
 where $\proj{\psi}$ is a purification of $\rho$ and where the supremum ranges over all density operators $\tau$ on the purifying system.  In particular, for any pure $\rho_{A B E} = \proj{\psi}$ we have
  \begin{align*}
    H_\alpha(A|B)_{\rho} = - \sup_{\tau_E} \frac{1}{\alpha'} \log  \left\| \bigl( \rho_B^{\frac{-\alpha'}{2}} \otimes \tau_E^{\frac{\alpha'}{2}} \bigr) \ket{\psi} \right\|^2 \ .
  \end{align*}
\end{lemma}

\begin{proof}
  Let us denote by $A$ the Hilbert space on which $\rho$ and $\sigma$ act and by $E$ the purifying space, so that $\ket{\psi}$ is a vector on $A \otimes E$. Then, using~\eqref{eq_alphatwo} and~\eqref{eq_oppart}, the sandwiched R\'enyi entropy can be written as
  \begin{align*}
    D_\alpha(\rho \| \sigma)
    & = \frac{1}{\alpha'} \log \left\| \sigma^{\frac{-\alpha'}{2}} \rho \sigma^{\frac{-\alpha'}{2}} \right\|_\alpha \\
    & = \frac{1}{\alpha'} \log \left\| \sigma^{\frac{-\alpha'}{2}} \mathrm{op}_{A \leftarrow E}(\ket{\psi}) \mathrm{op}_{A \leftarrow E}(\ket{\psi})^{\dagger} \sigma^{\frac{-\alpha'}{2}} \right\|_\alpha \\
        & = \frac{1}{\alpha'} \log \left\|  \mathrm{op}_{A \leftarrow E}(\ket{\psi})^{\dagger} \sigma^{\frac{-\alpha'}{2}} \sigma^{\frac{-\alpha'}{2}} \mathrm{op}_{A \leftarrow E}(\ket{\psi})  \right\|_\alpha \ .
  \end{align*}
  Using Lemma~\ref{lem_normdual} as well as~\eqref{eq_alphatwo} and \eqref{eq_opinc} we obtain
  \begin{align*}
      D_\alpha(\rho \| \sigma)
    & = \sup_{\tau}  \frac{1}{\alpha'} \log \tr\left(   \mathrm{op}_{A \leftarrow E}(\ket{\psi})^{\dagger} \sigma^{\frac{-\alpha'}{2}} \sigma^{\frac{-\alpha'}{2}} \mathrm{op}_{A \leftarrow E}(\ket{\psi})  \tau^{\alpha'} \right) \\
    & = \sup_{\tau}  \frac{1}{\alpha'} \log  \left\|  \sigma^{\frac{-\alpha'}{2}} \mathrm{op}_{A \leftarrow E}\bigl(\ket{\psi}\bigr)  \tau^{\frac{\alpha'}{2}} \right\|_2^2 \\
    & = \sup_{\tau}  \frac{1}{\alpha'} \log  \left\| \mathrm{op}_{A \leftarrow E}\bigl(\sigma^{\frac{-\alpha'}{2}} \otimes \tau^{\frac{\alpha'}{2}}  \ket{\psi}\bigr) \right\|_{2}^2 \ ,
      \end{align*}
  where the supremum is taken over density operators $\tau$ on $E$.  The first equality of the lemma then follows by~\eqref{eq_opnorm}. Finally, the second equality is obtained via~\eqref{eq_condrelentropy}. 
\end{proof}

The next lemma concerns the conditioning on classical information. 

\begin{lemma}[Proposition 5.1 of \cite{Tom15}] \label{lem_classicalsideinformation}
  For any density operator $\rho_{A B X}$ which is classical on $X$, i.e., 
  \begin{align*}
    \rho_{A B X} = \sum_{x \in \cX} p_x \, \rho_{A B | x} \otimes \proj{x}_X \ ,
  \end{align*}
  where $\rho_{A B | x}$ are density operators on $A \otimes B$ and $\{\ket{x}\}_{x \in \cX}$ is an orthonormal basis of $X$, we have
  \begin{align*}
    H_\alpha(A | B X)_{\rho} = \frac{1}{1-\alpha} \log \sum_{x \in \cX} p_x 2^{(1-\alpha) H_\alpha(A|B)_{\rho_{|x}}} \ .
  \end{align*}
\end{lemma}

\begin{proof}
  Using the explicit form of $\rho_{A B X}$, it is straightforward to verify that
  \begin{align*}
    \left(\rho_{B X}^{\frac{1-\alpha}{2 \alpha}} \rho_{A B X} \rho_{B X}^{\frac{1-\alpha}{2 \alpha}}\right)^\alpha
    = \sum_{x \in \cX} p_x  \left(\rho_{B|x}^{\frac{1-\alpha}{2 \alpha}} \rho_{A B|x} \rho_{B|x}^{\frac{1-\alpha}{2 \alpha}}\right)^\alpha \otimes \proj{x}_X \ .
  \end{align*}
  Taking the trace on both sides, the equality can be rewritten in terms of $\alpha$-entropies as
  \begin{align*}
    2^{(1-\alpha) H_\alpha(A | B X)_{\rho} } =  \sum_{x \in \cX} p_x  2^{(1-\alpha) H_\alpha(A | B)_{\rho_{|x}} } \ ,
  \end{align*}
  which concludes the proof. 
\end{proof}

  The following lemma can be found as Lemma 3.9 in \cite{MuellerLennertMaster}; the statement and its proof are given here for the convenience of the reader. The statement and proof can also be found in~\cite[Proposition 6.5]{Tom15}.
  
  \begin{lemma}[Lemma 3.9 from \cite{MuellerLennertMaster}, a variant of Proposition 6.2 of~\cite{Tom12}]\label{lem:dmin-vs-dalpha}
  Let $\rho \in \mathrm{D}_{\leqslant}(A)$ and $\sigma \in \Pos(A)$ with $\Supp(\rho) \subseteq \Supp(\sigma)$, and define $\varepsilon_{\max} := \sqrt{2 \tr \rho - (\tr \rho)^2}$. For $\varepsilon \in (0, \varepsilon_{\max})$ and $\alpha \in (1,2]$, we have
  \[ D_{\max}^{\varepsilon}(\rho \| \sigma) \leqslant D_{\alpha}(\rho \| \sigma) + \frac{g(\varepsilon)}{\alpha-1}, \]
  where $g(\varepsilon) = -\log\left( 1 - \sqrt{1 - \varepsilon^2} \right)$, and $D_{\max}^{\varepsilon}(\rho \| \sigma) = \inf_{\tilde{\rho}} \inf \{ \lambda : \tilde{\rho} \leqslant 2^{\lambda} \sigma \}$, with the infimum ranging over all $\tilde{\rho}$ within $\varepsilon$ of $\rho$ in purified distance.
  \end{lemma}
  
  \begin{proof}
      Assume without loss of generality that $\sigma$ has full support. By Lemma 6.1 in \cite{Tom12}, we can find a $\lambda$ such that $\lambda \geqslant D_{\max}^{\varepsilon}(\rho \| \sigma)$ where 
      \begin{equation}\label{eqn:mml-epsilon}
          \varepsilon = \sqrt{2 \tr[\Delta] - \tr[\Delta]^2}
      \end{equation}
      and $\Delta$ is the positive part of $\rho - 2^{\lambda} \sigma$.\footnote{The positive part of a Hermitian operator $X$ is defined as $\{X > 0\} X$, where $\{X > 0\}$ is the projector onto the span of the eigenspaces of $X$ with positive eigenvalues.} It suffices to upper-bound $\lambda$ by $D_{\alpha}(\rho \| \sigma) + g(\varepsilon)/(\alpha-1)$. Now, let $\{ \ket{e_i} \}_{i \in S}$ be an orthonormal basis consisting of eigenvectors of $\rho - 2^{\lambda} \sigma$. Let $S_+$ be the subset of $S$ corresponding to positive eigenvalues. Define the non-negative numbers $r_i = \bra{e_i} \rho \ket{e_i}$ and $s_i = \bra{e_i} \sigma \ket{e_i}$. Note that for $i \in S_{+}$, we have $r_i - 2^{\lambda} s_i = \bra{e_i} (\rho - 2^{\lambda} \sigma) \ket{e_i} \geqslant 0$ and therefore $\frac{r_i}{s_i} 2^{-\lambda} \geqslant 1$. We use this to bound
      \begin{align*}
          \tr[ \Delta ] &= \sum_{i \in S_+} (r_i - 2^{\lambda} s_i) \leqslant \sum_{i \in S_+} r_i \leqslant \sum_{i \in S_+} r_i \left( \frac{r_i}{s_i} 2^{-\lambda} \right)^{\alpha - 1}\\
          &= 2^{-\lambda(\alpha-1)} \sum_{i \in S_+} r_i^{\alpha} s_i^{1-\alpha} \leqslant 2^{-\lambda(\alpha-1)} \sum_{i \in S} r_i^{\alpha} s_i^{1-\alpha}.
      \end{align*}
      Hence,
      \[ \frac{1}{\alpha-1} \log \tr[ \Delta ] \leqslant \frac{1}{\alpha-1} \log \sum_{i \in S} r_i^{\alpha} s_i^{1-\alpha} - \lambda. \]
      Now, we solve Equation (\ref{eqn:mml-epsilon}) for $\tr[\Delta]$ and bound
      \[ \lambda \leqslant \frac{1}{\alpha-1} \log \sum_{i \in S} r_i^{\alpha} s_i^{1-\alpha} - \frac{1}{\alpha-1} \log\left( 1 - \sqrt{1-\varepsilon^2} \right). \]
      It remains to upper-bound $\frac{1}{\alpha-1} \log \sum_{i \in S} r_i^{\alpha} s_i^{1-\alpha}$ by $D_{\alpha}(\rho \| \sigma)$. To this end, we define the TPCP map $\mathcal{F}(X) = \sum_{i \in S} P_i X P_i$, where $P_i$ denotes the projector onto the subspace spanned by $e_i$. Note that
      \[ D_{\alpha}\left( \mathcal{F}(\rho) \| \mathcal{F}(\sigma) \right) = \frac{1}{\alpha-1} \log \sum_{i \in S} r_i^{\alpha} s_i^{1-\alpha}. \]
      The theorem then follows from the data processing inequality.
  \end{proof}

The following two lemmas relate the entropy conditioned on a classical value $x$ to the unconditioned entropy.

\begin{lemma}\label{lem:abx-chain-rule-opt}
    Let $\rho_{AB}$ be a quantum state of the form $\rho = \sum_x p_x {\rho}_{AB|x}$, where $\{ p_x \}$ is a probability distribution over $\mathcal{X}$. Then, for any $x \in \mathcal{X}$ and any $\alpha \in (1, \infty)$,
    \begin{align}
        \label{eqn:abx-chain-rule-gtr1-opt} H^{\uparrow}_\alpha(A|B)_{\rho} - \frac{\alpha}{\alpha - 1} \log \left( \frac{1}{p_x} \right)  &\leqslant  H^{\uparrow}_{\alpha}(A|B)_{\rho_{|x}} \ .
    \end{align}
    and for $\alpha \in (0,1)$, 
        \begin{align}
        \label{eqn:abx-chain-rule-gtr1-opt-lessthan1} H^{\uparrow}_\alpha(A|B)_{\rho} - \frac{\alpha}{\alpha - 1} \log \left( \frac{1}{p_x} \right)  &\geqslant  H^{\uparrow}_{\alpha}(A|B)_{\rho_{|x}} \ .
    \end{align}
\end{lemma}
\begin{proof}
For any $\sigma_{B}$ and $\alpha \in (1,\infty)$, we have
\begin{align*}
D_{\alpha}(\rho_{AB} \| \id_{A} \otimes \sigma_{B}) 
&= \frac{1}{\alpha - 1} \log \tr\left( \left(\sigma_B^{\frac{1-\alpha}{2\alpha}} \rho_{AB} \sigma_B^{\frac{1-\alpha}{2\alpha}} \right)^{\alpha} \right)\\
&\geqslant \frac{1}{\alpha - 1} \log \tr\left( \left(\sigma_B^{\frac{1-\alpha}{2\alpha}} p_x \rho_{AB|x} \sigma_B^{\frac{1-\alpha}{2\alpha}} \right)^{\alpha} \right)\\
&= \frac{1}{\alpha - 1} \log p_x^{\alpha} + \frac{1}{\alpha - 1} \log \tr\left( \left(\sigma_B^{\frac{1-\alpha}{2\alpha}} \rho_{AB|x} \sigma_B^{\frac{1-\alpha}{2\alpha}} \right)^{\alpha} \right)\\
&= \frac{\alpha}{\alpha - 1} \log p_x + D_{\alpha}(\rho_{AB|x} \| \id_{A} \otimes \sigma_B) \ .
\end{align*}
For the first inequality, we used the fact that $\rho_{AB} = \sum_{x'} p_{x'} \rho_{AB|x'} \geqslant p_x \rho_{AB|x}$, which implies that $\sigma_B^{\frac{1-\alpha}{2\alpha}} \rho_{AB} \sigma_B^{\frac{1-\alpha}{2\alpha}} \geqslant \sigma_B^{\frac{1-\alpha}{2\alpha}} p_x \rho_{AB|x} \sigma_B^{\frac{1-\alpha}{2\alpha}}$. We then used the fact that $y \mapsto y^{\alpha}$ is a monotone function on $[0, \infty)$. Taking the infimum over $\sigma_B$ and then multiplying both sides by $-1$, we get the desired result. The proof is the same for $\alpha \in (0,1)$ except that the direction of the inequality is reversed.
\end{proof}

\begin{lemma}\label{lem:abx-chain-rule}
    Let $\rho_{AB}$ be a quantum state of the form $\rho = \sum_x p_x {\rho}_{AB|x}$, where $\{ p_x \}$ is a probability distribution over $\mathcal{X}$. Then, for any $x \in \mathcal{X}$ and any $\alpha \in (1, 2]$,
    \begin{align}
        \label{eqn:abx-chain-rule-less1} H_{\frac{1}{\alpha}}(A|B)_{\rho} + \frac{\alpha}{\alpha - 1} \log \left( \frac{1}{p_x} \right) &\geqslant H_{\frac{1}{\alpha}}(A|B)_{\rho_{|x}}.
    \end{align}
\end{lemma}
\begin{proof}
    We define the state $\rho_{ABX} = \sum_{x'} p_{x'} \rho_{AB|x'} \otimes \proj{x'}_X$. Note that it is legitimate to use the notation $\rho$ as the reduced state on $A \otimes B$ corresponds to $\rho_{AB}$. As conditioning can only decrease the entropy, we obtain
    \begin{align*}
        H_{\frac{1}{\alpha}}(A|B)_{\rho} 
        &\geqslant H_{\frac{1}{\alpha}}(A|BX)_{\rho}\\
        &= \frac{\alpha}{\alpha - 1} \log \tr\left[ \left( \rho_{BX}^{\frac{\alpha-1}{2}} \rho_{ABX} \rho_{BX}^{\frac{\alpha-1}{2}}\right)^{\frac{1}{\alpha}} \right]\\
        &= \frac{\alpha}{\alpha-1} \log \sum_{x'} p_{x'} \tr\left[ \left( \rho_{B|x'}^{\frac{\alpha-1}{2}} \rho_{AB|x'} \rho_{B|x'}^{\frac{\alpha-1}{2}}\right)^{\frac{1}{\alpha}} \right]\\
        &\geqslant \frac{\alpha}{\alpha-1} \log p_{x} \tr\left[ \left( \rho_{B|x}^{\frac{\alpha-1}{2}} \rho_{AB|x} \rho_{B|x}^{\frac{\alpha-1}{2}}\right)^{\frac{1}{\alpha}} \right]\\
        &= H_{\frac{1}{\alpha}}(A|B)_{\rho_{|x}} - \frac{\alpha}{\alpha-1} \log \left( \frac{1}{p_x} \right)
   \ .
    \end{align*}
\end{proof}

\begin{lemma}
\label{lem_det_func}
  Let $\cE$ be a TPCP map from $A \otimes B$ to $A \otimes B \otimes X$ defined by $\cE(W_{AB}) = \sum_{y,z} (\Pi_{y,A} \otimes \Pi_{z,B}) W_{AB} (\Pi_{y,A} \otimes \Pi_{z,B}) \otimes \proj{t(y,z)}_{X}$, where $t : \cY \times \cZ \to \cX$ is a (deterministic) function,  $\{\Pi_{y,A}\}_{y \in \cY}$ and $\{\Pi_{z, B}\}_{z \in \cZ}$  are mutually orthogonal projectors acting on $A$ and $B$, respectively, and $\{\ket{x}\}_{x \in \cX}$ is an orthonormal basis on $X$. Let $\rho_{ABX} = \cE(\omega_{AB})$, for an arbitrary state $\omega_{AB}$. Then for $\alpha \in [\frac12, \infty)$, we have
   \begin{align}
        H^{\uparrow}_\alpha(AX|B)_{\rho} &= H^{\uparrow}_{\alpha}(A|B)_{\rho} \ , \label{eq_det_func_uparrow} \\
        H_\alpha(AX|B)_{\rho} &= H_{\alpha}(A|B)_{\rho} \label{eq_det_func_noarrow} \ .
    \end{align}
\end{lemma}
\begin{proof}
We only prove Eq.~\eqref{eq_det_func_uparrow}. Eq.~\eqref{eq_det_func_noarrow} is easier.
Let $\cM$ be the TPCP map from $B$ to $B$ defined by $\cM(W_B) = \sum_{z} \Pi_{z, B} W_{B} \Pi_{z, B}$. Using the data processing inequality and the fact that $(\cI_{AX} \otimes \cM )(\rho_{ABX}) = \rho_{ABX}$, we have
\begin{align}
H^{\uparrow}_\alpha(AX|B)_{\rho} = - \inf_{\sigma} D_{\alpha}(\rho_{ABX} \| \id_{AX} \otimes \sigma_B) = - \inf_{\sigma} D_{\alpha}(\rho_{ABX} \| \id_{AX} \otimes \cM(\sigma_B)) \ .
\end{align}
Similarly,
\begin{align}
H^{\uparrow}_\alpha(A|B)_{\rho} = - \inf_{\sigma} D_{\alpha}(\rho_{AB} \| \id_{A} \otimes \cM(\sigma_B)) \ .
\end{align}
We now show that for any state $\sigma_B$, we have $D_{\alpha}(\rho_{ABX} \| \id_{AX} \otimes \cM(\sigma_B)) = D_{\alpha}(\rho_{AB} \| \id_{A} \otimes \cM(\sigma_B))$. To make the notation lighter, we use in the following $\Pi_{z}$ for $\Pi_{z,B}$ and $\Pi_{y}$ for $\Pi_{y,A}$. The relative entropy $D_{\alpha}(\rho_{ABX} \| \id_{AX} \otimes \cM(\sigma_B))$ is defined in terms of
\begin{align*}
&\tr \left( \cM(\sigma_{B})^{\frac{1-\alpha}{2\alpha}} \rho_{ABX} \cM(\sigma_{B})^{\frac{1-\alpha}{2\alpha}} \right)^{\alpha} \\
&=  \tr \left( \sum_{y,z} (\Pi_{z} \sigma_{B} \Pi_{z})^{\frac{1-\alpha}{2\alpha}} (\Pi_{y} \otimes \Pi_{z}) \omega_{AB} (\Pi_{y} \otimes \Pi_{z}) \otimes \proj{t(y,z)} (\Pi_{z} \sigma_{B} \Pi_{z})^{\frac{1-\alpha}{2\alpha}} \right)^{\alpha} \\
&= \sum_{y,z} \tr   \left( (\Pi_{z} \sigma_{B} \Pi_{z})^{\frac{1-\alpha}{2\alpha}} (\Pi_{y} \otimes \Pi_{z}) \omega_{AB} (\Pi_{y} \otimes \Pi_{z}) \otimes \proj{t(y,z)} (\Pi_{z} \sigma_{B} \Pi_{z})^{\frac{1-\alpha}{2\alpha}} \right)^{\alpha} \\
&= \sum_{y,z} \tr \left( (\Pi_{z} \sigma_{B} \Pi_{z})^{\frac{1-\alpha}{2\alpha}} (\Pi_{y} \otimes \Pi_{z}) \omega_{AB} (\Pi_{y} \otimes \Pi_{z}) (\Pi_{z} \sigma_{B} \Pi_{z})^{\frac{1-\alpha}{2\alpha}} \right)^{\alpha} \ ,
\end{align*}
where we used multiple times the orthogonality of the family $\{\Pi_{z}\}$ and of the family $\{\Pi_{y}\}$. Similarly, $D_{\alpha}(\rho_{AB} \| \id_{A} \otimes \cM(\sigma_B))$ is defined in terms of
\begin{align*}
&\tr \left( \cM(\sigma_{B})^{\frac{1-\alpha}{2\alpha}} \rho_{AB} \cM(\sigma_{B})^{\frac{1-\alpha}{2\alpha}} \right)^{\alpha} \\
&= \sum_{y,z} \tr \left( (\Pi_{z} \sigma_{B} \Pi_{z})^{\frac{1-\alpha}{2\alpha}} (\Pi_{y} \otimes \Pi_{z}) \omega_{AB} (\Pi_{y} \otimes \Pi_{z}) (\Pi_{z} \sigma_{B} \Pi_{z})^{\frac{1-\alpha}{2\alpha}} \right)^{\alpha} \ .
\end{align*}
And this concludes the proof.
\end{proof}

In the subsequent arguments we will use the quantity
    \begin{align*}
        D'_\alpha(\rho\|\sigma) = \frac{1}{\alpha-1} \log\tr(\rho^\alpha \sigma^{1-\alpha}) \ .
  \end{align*}
  which is defined for any non-negative operators $\rho$ and $\sigma$ on the same space and for any $\alpha \in [0, 1) \cup (1, \infty)$. As observed in~\cite{WWY13,Datta,FrankLieb}, it follows from the Araki-Lieb-Thirring inequality~\cite{LiebThirring,Araki} that
  \begin{align} \label{eq_DDprimerelation}
    D_\alpha(\rho \| \sigma) \leqslant D'_{\alpha}(\rho \| \sigma) \ .
  \end{align}
  Furthermore, we can define a conditional entropy based on this quantity:
  \begin{equation}\label{eqn:halpha-duality}
      H'_{\alpha}(A|B)_{\rho} := - \inf_{\sigma_B} D'(\rho_{AB} \| \ident_A \otimes \sigma_B).
  \end{equation}
  In \cite[Theorem 2]{TomamichelBertaHayashi}, it is shown that $H'$ and $H$ are duals of each other, in the sense that
  \begin{align} \label{eq_Hduality}
   H_{\alpha}(A|B)_{\rho} = -H'_{\frac{1}{\alpha}}(A|C)_{\rho} 
  \end{align}
  for any pure state $\rho_{ABC}$.

The following lemma is another variant of Lemma~8 of~\cite{TCR09} (see also Lemma~6.3 of~\cite{Tom12}). 

\begin{lemma} \label{lem_DalphaDbound}
  Let $\rho$ be a density operator and $\sigma$ a non-negative operator, let $\eta = \max(4, 2^{D'_2(\rho \| \sigma)} + 2^{-D'_0(\rho \| \sigma)} + 1)$, and let $\alpha \in (1, 1 + 1/ \log \eta) $. Then
  \begin{align*}
    D'_{\alpha}(\rho \| \sigma) < D(\rho \| \sigma) + (\alpha-1) (\log \eta)^2 \ .
  \end{align*}
\end{lemma}

\begin{proof}[Proof sketch.]
  The proof proceeds in the same way as the proof of Lemma~8 of~\cite{TCR09}. The idea is to consider, for any $\beta > 0$, the functions $r_{\beta}$ and $s_{\beta}$ from $\mathbb{R}^+$ to $\mathbb{R}^+$ defined by
  \begin{align*}
    r_{\beta}(t) & = t^\beta - \beta \ln t - 1 \\
    s_{\beta}(t) & = 2 (\cosh(\beta \ln t) - 1)  \ .
  \end{align*}
  It can be readily verified that $r_\beta(t) \leqslant s_\beta(t)$ for all $t > 0$, that $s_{\beta}(t) = s_{\beta}(1/t)$, that $s_{\beta}(t)$ is monotonically increasing for $t > 1$, and that $s_{\beta}(t)$ is concave for $\beta < 1/2$ and $t \geqslant 3$. 
  It is then shown that
  \begin{align*}
    D'_{\alpha}(\rho \| \sigma) \leqslant D(\rho \| \sigma) + \frac{1}{\beta \ln 2} \bra{\phi} r_{\beta}(X) \ket{\phi} \ , 
  \end{align*}
  where $\beta = \alpha - 1$, $X = \rho \otimes \mathrm{\sigma^{-1}}^T$, and $\ket{\phi} = (\sqrt{\rho} \otimes \id) \ket{\gamma}$ with  $\ket{\gamma} = \sum_i \ket{i} \otimes \ket{i}$. 
  
  From there, we proceed in a slightly different way, noting that
  \begin{align*}
    s_{\beta}(t) = s_{\beta}(\max(t, \frac{1}{t})) \leqslant s_{\beta}(\max(t, \frac{1}{t})+1) \leqslant s_{\beta}(t+\frac{1}{t} +1) \ .
  \end{align*}
  Using this as well as Lemma~11 of~\cite{TCR09}, we obtain
  \begin{align*}
    \bra{\phi} r_{\beta}(X) \ket{\phi} \leqslant \bra{\phi} s_{\beta}(X) \ket{\phi} 
    \leqslant  \bra{\phi} s_{\beta}(X + \frac{1}{X} + \mathbb{I}) \ket{\phi} 
    \leqslant  s_{\beta}(\bra{\phi} X + \frac{1}{X} + \mathbb{I}\ket{\phi})  \ ,
  \end{align*}
  which holds because $s_\beta$ is concave for $\beta \leqslant \frac{1}{\log \eta} \leqslant \frac{1}{2}$, and because the eigenvalues of $X+\frac{1}{X} + \mathbb{I}$ lie in the interval $[3, \infty)$. Using that 
  \begin{align*}
    \bra{\phi} X+\frac{1}{X} + \mathbb{I} \ket{\phi} = 2^{D'_2(\rho \| \sigma)} + 2^{-D'_0(\rho \| \sigma)} + 1 \leqslant \eta
  \end{align*}
  and combining the inequalities above, we find
  \begin{align*}
      D'_{\alpha}(\rho \| \sigma) \leqslant D(\rho \| \sigma) + \frac{1}{\beta \ln 2} s_{\beta}(\eta) \ .
  \end{align*}
  Applying Taylor's theorem to an expansion around $\beta = 0$ gives
  \begin{align*}
    \frac{1}{\beta \ln 2} s_\beta(\eta) \leqslant \frac{(\beta \ln \eta)^2 \cosh(\beta \ln \eta)}{\beta \ln 2} \leqslant \beta (\log \eta)^2 \ln 2 \cosh(\beta \ln \eta) < \beta (\log \eta)^2 \ ,
  \end{align*}
  where for the last inequality, we use the fact that $\ln 2 \cosh(\ln 2) < 1$.
\end{proof}

The following lemma is a generalisation of Proposition~3.10 of~\cite{MuellerLennertMaster}. In~\cite[Section 6.4.2]{Tom15}, using a Taylor approximation, the factor in front of $(\alpha - 1)$ can be improved although at the price of having the error term containing non-explicit constants.

\begin{lemma} \label{lem_HalphaH}
   For any density operator $\rho_{A B}$ and $1 < \alpha < 1+ 1/\log(1 + 2 d_A)$
   \begin{align*}
       H(A|B)_{\rho} -   (\alpha - 1) \log^2(1+2 d_A) &<  H_\alpha(A|B)_{\rho} \\
        &\leqslant H_{\frac{1}{\alpha}}(A|B)_{\rho} <  H(A|B)_{\rho} +  (\alpha - 1) \log^2(1 + 2 d_A) \ ,
  \end{align*}
  where $d_A = \dim A$. 
  \end{lemma}

\begin{proof}
 We start with the proof of the first inequality. Lemma~\ref{lem_DalphaDbound} implies that
  \begin{align*}
      H(A|B)_{\rho} - (\alpha - 1)(\log (1+2 d_A))^2 <  -D'_\alpha(\rho_{A B} \| \id_A \otimes \rho_B) 
  \end{align*}
  holds for all $1 < \alpha < 1 + \frac{1}{\log(1 + 2 d_A)}$, where we have used that
   \begin{align*}
   1 + 2 d_A \geqslant 1 + 2^{-D'_{0}(\rho_{A B}\| \id_A \otimes \rho_B)} + 2^{D'_{2}(\rho_{A B}\| \id_A \otimes \rho_B)} \ . 
  \end{align*}
  Furthermore, because of~\eqref{eq_DDprimerelation} we have
  \begin{align*}
    - D'_{\alpha}(\rho_{A B} \| \id_A \otimes \rho_B) 
    \leqslant -D_\alpha(\rho_{A B} \| \id_A \otimes \rho_B)
    = H_\alpha(A|B)_{\rho}  \ .
  \end{align*}
  Combining this with the above concludes the proof of the first inequality. 
  
  The second inequality follows directly from the monotonicity of the relative R\'enyi entropy in $\alpha$~\cite{Beigi,MDSFT13}. 
  
%
  
    To prove the last inequality, we again use the duality relation (\ref{eqn:halpha-duality}):
    \begin{align*}
        H_{\frac{1}{\alpha}}(A | B)_{\rho} &= -H'_{\alpha}(A|C)_{\rho}\\
        &= \inf_{\sigma} D'_{\alpha}(\rho_{A C} \| \id_A \otimes \sigma_C)\\
        &\leqslant D'_{\alpha}(\rho_{A C} \| \id_A \otimes  \rho_C) \ .
  \end{align*}
  We may now again use  Lemma~\ref{lem_DalphaDbound} to obtain
  \begin{align*}
    D'_{\alpha}(\rho_{A C} \| \id_A \otimes \rho_C) < - H(A|C)_{\rho} + (\alpha - 1)  (\log(1 + 2 d_A))^2 \ .
  \end{align*}
  Combining the two inequalities with the fact that $-H(A|C) = H(A|B)$ concludes the proof.   
  
\end{proof}

The following lemma generalises a classical result originally proposed in~\cite{RennerWolf}. It follows rather directly from similar statements proved in~\cite{TCR09,Tom12,MuellerLennertMaster,MS14a}.

\begin{lemma} \label{lem_Hepsalpha}
  For any density operator $\rho$,  any non-negative operator $\sigma$, any $\alpha \in (1, 2]$, and any $0 < \eps < 1$, 
  \begin{align*}
    H_{\min}^{\eps}(A|B)_{\rho} & \geqslant H^{\uparrow}_{\alpha}(A|B)_{\rho} -  \frac{g(\eps)}{\alpha - 1}  \\
    H_{\max}^{\eps}(A|B)_{\rho} & \leqslant H_{\frac{1}{\alpha}}(A|B)_{\rho} +  \frac{g(\eps)}{\alpha - 1} 
   \ ,
  \end{align*}
  where $g(\eps) = - \log (1 - \sqrt{1-\eps^2}) < \log(2/\eps^2)$. 
\end{lemma}

\begin{proof}
    For the first inequality, we use Lemma~\ref{lem:dmin-vs-dalpha}, which directly implies that
 \begin{align*}
     H_{\min}^\eps(A | B)_{\rho} \geqslant \sup_{\sigma_B} \big( - D_\alpha(\rho_{A B} \| \id_A \otimes \sigma_B) \big) -   \frac{g(\eps)}{\alpha - 1} \ . 
 \end{align*}
 The desired inequality then follows because
 \begin{align*}
     \sup_{\sigma_B} \big( - D_\alpha(\rho_{A B} \| \id_A \otimes \sigma_B) \big)  = H^{\uparrow}_\alpha(A | B)_{\rho} \ .
 \end{align*}
 
 To prove the second inequality we use the duality between smooth min- and max-entropy~\cite{TCR10}, which asserts that
\begin{align*}
  H_{\max}^\eps(A|B)_{\rho}
   = - H_{\min}^\eps(A|C)_{\rho} 
\end{align*}
hols for any purification $\rho_{A B C}$ of $\rho_{A B}$.  We can then employ Proposition~6.2 of~\cite{Tom12},\footnote{Alternatively, one may again use Lemma~3.9 of~\cite{MuellerLennertMaster}, which asserts that the relation is also true for $D_\alpha$ instead of $D'_\alpha$, together with~\eqref{eq_DDprimerelation}.}
 \begin{align*}
   - H_{\min}^\eps(A|C)_{\rho}  \leqslant \inf_{\sigma_C} D'_{\alpha}(\rho_{A C}\|\id_A \otimes \sigma_C) + \frac{g(\eps)}{\alpha - 1} \ .
 \end{align*}
 The claim then follows from (\ref{eqn:halpha-duality}):
 \begin{align*}
     \inf_{\sigma_C} D'_{\alpha}(\rho_{A C} \| \id_A \otimes \sigma_C) &= -H'_{\alpha}(A|C)_{\rho}\\
     &= H_{\frac{1}{\alpha}}(A|B)_{\rho} \ .
\end{align*}
\end{proof}

\section{Necessity of the Markov chain conditions}  \label{app_Markov}

The aim of this section is to illustrate that the Markov chain conditions in Theorem~\ref{thm:entropyaccumulationext} are important, in the sense that dropping them completely would render the statement invalid.

We first recall that a tri-partite density operator $\rho_{A B C}$ has the Markov state property $A \leftrightarrow B \leftrightarrow C$  if and only if the mutual information between $A$ and $C$ conditioned on $B$ equals zero, i.e., ${I(A : C | B)_{\rho}} = 0$~(see~\cite{Pet88,HJPW04}, as well as~\cite{FR2015} for a robust version). Using the properties of the conditional mutual information, one can easily derive the following claims:
\begin{itemize}
  \item Symmetry: $A \leftrightarrow B \leftrightarrow C$ implies $C \leftrightarrow B \leftrightarrow A$
  \item Local processing of endpoints: $A A' \leftrightarrow B \leftrightarrow C$ implies $A \leftrightarrow B \leftrightarrow C$
  \item Centering of information: $A A' \leftrightarrow B \leftrightarrow C$ implies $A \leftrightarrow A' B \leftrightarrow C$
  \item Composition: $A \leftrightarrow B \leftrightarrow C$ and $A' \leftrightarrow A B \leftrightarrow C$  imply $A A' \leftrightarrow B \leftrightarrow C$.
\end{itemize}

By the standard properties of Markov chains described above, it is straightforward to show that the set of Markov chain conditions~\eqref{eq_Markovgen} for a trivial $E$ system is equivalent to the set of conditions
\begin{equation}\label{eqn:Markovcond}
  A_1^{i} \leftrightarrow B_1^{i} \leftrightarrow B_{i+1}^{n} \qquad (i = 1, \ldots, n-1).
\end{equation}
Similarly, using the composition property, one can show that this set of conditions is equivalent to the set of conditions
\begin{align*}
  A_i \leftrightarrow A_1^{i-1} B_1^i \leftrightarrow B_{i+1}^n \qquad (i = 1, \ldots, n-1).
\end{align*}
The latter can be expressed in terms of the entropy equalities
\begin{align*}
   H(A_i | A_1^{i-1} B_1^i)_{\rho} = H(A_i | A_1^{i-1} B_1^n)_{\rho} \ .
\end{align*}
The entropy accumulation statement for the smooth min-entropy in the simplified form~\eqref{eq_entropyaccumulation} can thus be rewritten as
\begin{align} \label{eq_entropyaccumulationre}
     H_{\min}^\eps(A_1^n | B_1^n)_{\rho} > \sum_{i =1}^n \inf_{\omega} H(A_i | A_1^{i-1} B_1^{n})_{\cM_i(\omega)} - c_\eps \sqrt{n} \ .
\end{align}
Note that, if one replaces the  smooth min-entropy on the left hand side by the von Neumann entropy then this expression looks similar to the usual chain rule for von Neumann entropies, 
\begin{align*}
    H(A_1^n | B_1^n)_\rho = \sum_{i =1}^n H(A_i | A_1^{i-1} B_1^{n})_\rho  \ ,
\end{align*}
which holds for arbitrary $\rho_{A_1^n B_1^n}$. One may therefore wonder whether~\eqref{eq_entropyaccumulationre} may also hold without the Markov conditions~\eqref{eqn:Markovcond}. This is however not the case, as we are going to show with a specific example. 

The example is classical, in the sense that $A_1, \ldots, A_n$ and $B_1, \ldots B_n$ correspond to random variables and the map $\cM_i$ takes $a_1^{i-1} b_1^{i-1}$ as input and outputs $a_1^{i-1} b_1^{i-1}$ as is, together with $A_i B_i$ generated from the conditional distribution $\rho_{A_i B_i|A_1^{i-1} = a_1^{i-1}, B_1^{i-1} = b_1^{i-1}}$. With this setup, \eqref{eq_entropyaccumulationre} can be written as
  \begin{align}  \label{eq_maininequalitygeneralisedclassical}
     H_{\min}^\eps(A_1^n | B_1^n)_{\rho} > \ \sum_i \inf_{a_1^{i-1}, b_1^{i-1}, b_{i+1}^n} H(A_i|B_i)_{\rho_{|a_1^{i-1} b_1^{i-1} b_{i+1}^n}} -  c_\eps \sqrt{n}  \ .
\end{align}
Actually, \eqref{eq_maininequalitygeneralisedclassical} is even weaker than what would follow from~\eqref{eq_entropyaccumulationre} as we are taking the infimum also over $b_{i+1}^{n}$ but we will see that even this weaker inequality is false. In fact, consider the following construction, let $B_1, B_2, \ldots, B_n$ be $n$ mutually independent and uniformly distributed  $n$-bit strings. Furthermore let $C$ be a uniform random bit and let $A = A_1^n$ be an $n$-bit-string defined by
  \begin{align*}
    A = \begin{cases} B_1 \oplus B_2 \oplus \cdots \oplus B_n & \text{if $C=0$} \\ \text{uniform and independent} & \text{if $C=1$,} \end{cases} 
  \end{align*}
  where $\oplus$ denotes the bit-wise addition modulo~$2$. 
  In other words, with probability $1/2$,  the string $A$ is fully determined by $B_1, \dots, B_n$, and with probability $1/2$, $A$ is completely random. We then have, for $\eps \ll 1$,
  \begin{align*}
    H_{\min}^\eps(A_1^n | B_1^n)_{\rho} = H_{\min}^\eps(A_1^n | B_1 \oplus B_2 \oplus \cdots \oplus B_n)_{\rho}  \approx 1 \ .
  \end{align*}
  Furthermore, for any $i$ and for any fixed $a_1^{i-1}$, $b_1^{i-1}$, and $b_{i+1}^n$, we have 
  \begin{align*}
    H(A_i|B_i)_{\rho_{|a_1^{i-1} b_1^{i-1} b_{i+1}^{n}}} \geqslant \frac{1}{2} \ ,
  \end{align*}
  for the bit $A_i$ is random with probability $1/2$.  Since there are $n$ such terms in the sum on the right hand side of~\eqref{eq_maininequalitygeneralisedclassical}, it scales linearly in $n$. But we have just seen that the left hand side is roughly equal to~$1$. This shows that this inequality, and hence also the quantum version~\eqref{eq_entropyaccumulationre}, cannot hold in general if we drop the Markov chain conditions~\eqref{eqn:Markovcond}. 

\end{appendices}

\section*{Acknowledgments}
The authors would like to thank Martin M\"uller-Lennert, for allowing us to recopy Lemma~3.9 from his Master's thesis into this paper (as Lemma~\ref{lem:dmin-vs-dalpha}).  We also thank Rotem Arnon-Friedman and Thomas Vidick for useful discussions about applications of the entropy accumulation theorem, Carl Miller for discussions on randomness expansion and Marco Toma\-michel as well as the anonymous reviewers for comments on the manuscript. FD acknowledges the financial support of the Czech Science Foundation (GA {\v C}R) project no GA16-22211S and of the European Commission FP7 Project RAQUEL (grant No.\ 323970). OF acknowledges support from the French National Research Agency
via Project No. ANR-18-CE47-0011 (ACOM) and from LABEX MILYON (ANR-10-LABX-0070) of Université de Lyon, within the program ``Investissements d'Avenir" (ANR-11-IDEX-0007).
 RR acknowledges funding from the RAQUEL project, from the Swiss National Science Foundation (via grant No.\ 200020-135048 and the National Centre of Competence in Research ``Quantum Science and Technology''), from the European Research Council (grant No.\ 258932), and from the US Air Force Office of Scientific Research (grant Nos.\ FA9550-16-1-0245 and~FA9550-19-1-0202). 

\bibliographystyle{abbrv}
\bibliography{big2,biblio}

\end{document}